\documentclass[journal]{IEEEtran}
\usepackage{tabularx}

\ifCLASSINFOpdf
\else
   \usepackage[dvips]{graphicx}
\fi

 \makeatletter
\newcommand\notsotiny{\@setfontsize\notsotiny\@vipt\@viipt}
\makeatother

\usepackage{cite}
\usepackage{amsmath,amssymb,amsfonts,amsthm}
\usepackage{graphicx}
\usepackage{textcomp}

\usepackage{xcolor}
\usepackage[colorlinks=true,  linkcolor=black, citecolor=blue, breaklinks=true, urlcolor=blue]{hyperref}

\usepackage{graphicx}
\usepackage{thmtools}
\usepackage[ruled,vlined,linesnumbered]{algorithm2e}
\SetKwRepeat{Do}{do}{while}%

\usepackage{xcolor}
\usepackage{booktabs}
\usepackage{url}
\usepackage[latin1]{inputenc}

\definecolor{darkgreen}{rgb}{0.2, 0.6, 0.05}
\usepackage{enumitem}

\usepackage{subcaption}

\newcounter{ctheorem}
\newtheorem{theorem}[ctheorem]{Theorem}

\newcounter{clemma}
\newtheorem{lemma}[clemma]{Lemma}

\usepackage{pgfplots}
\usepackage{tikz}
\usepgfplotslibrary{groupplots}
\pgfplotsset{compat=1.13}
\newlength\fwidth

\usepackage{algorithmic}
\let\oldnl\nl
\newcommand{\nonl}{\renewcommand{\nl}{\let\nl\oldnl}}

\usetikzlibrary{arrows,petri,topaths}

\newcommand{\E}{\mathbb{E}}                  
\newcommand{\bo}[1]{\mathbf{#1}}              
\newcommand{\bom}[1]{\boldsymbol{#1}}    
\newcommand{\tildeb}{\mathbf{c}}

\newcommand{\bmu}{\boldsymbol{\mu}}

\newcommand{\supp}{\mathsf{supp}}
\newcommand{\pdim}{M}

\newcommand{\ns}{Q} 
\newcommand{\Gam}{\boldsymbol{\Gamma}}

\renewcommand{\Pi}{{\mathcal P}}
\newcommand{\hop}{\mathsf{H}}        

\newcommand{\gfnc}{g}  
\newcommand{\ffnc}{f} 

\newcommand{\beq}{\begin{equation}}
\newcommand{\eeq}{\end{equation}}
\newcommand{\bmat}{\begin{pmatrix}}
\newcommand{\emat}{\end{pmatrix}}

\def\Gam{\mathbf{\Gamma}}

\newcommand{\I}{\bo I} 
\newcommand{\R}{\mathbb{R}}

\newcommand{\X}{{\bf X}}
\newcommand{\SCM}{\mathbf{S}}

\newcommand{\A}{\mathbf{A}}

\newcommand{\Y}{\mathbf{Y}}
\newcommand{\mE}{\mathbf{E}}

\newcommand{\B}{{\bo B}}
\renewcommand{\S}{\mathbf{S}} 

\newcommand{\kdim}{K}
\newcommand{\sigmaw}{\sigma^2}

\newcommand{\C}{\mathbb{C}} 
\newcommand{\rsupp}{\mathsf{supp}} 

\newcommand{\y}{\bo y}

\renewcommand{\b}{\mathbf{b}}

\newcommand{\x}{\bo x}

\newcommand{\e}{\mathbf{e}}

\renewcommand{\a}{{\bo a}}
\renewcommand{\c}{\mathbf{c}}
\newcommand{\p}{{\bo p}}

\newcommand{\M}{\bom \Sigma}

\DeclareMathOperator{\tr}{tr}
\DeclareMathOperator{\Tr}{tr}
\DeclareMathOperator{\diag}{diag}
\DeclareMathOperator{\cov}{cov}

\newcommand{\bm}[1]{\mathbf{#1}}   
 \newcommand{\Hmat}{\mathsf{H}}

\newcommand{\gam}{\boldsymbol{\gamma}}
\newcommand{\iter}{k}

\begin{document}

\title{Efficient SCA Algorithms for Sparse Bayesian Learning}

\author{Esa~Ollila,~\IEEEmembership{Senior Member},~IEEE, Majdoddin Esfandiari,~\IEEEmembership{Member,~IEEE}, Daniel P. Palomar,~\IEEEmembership{Fellow,~IEEE}, and Xinjue Wang,~\IEEEmembership{Member,~IEEE}
\thanks{Esa Ollila, Majdoddin Esfandiari and Xinjue Wang are with the Department of Information and Communications Engineering, Aalto University, FI-00076 Aalto, Finland (e-mail: \{esa.ollila,majdoddin.esfandiari\}@aalto.fi, xjw@ieee.org).  Daniel P. Palomar is with the Hong Kong University of Science and Technology, Hong Kong (e-mail: palomar@ust.hk). The work of M. Esfandiari was supported by the Research Council of Finland (RCF) under Grant 368630,  D. Palomar in part by the Hong Kong GRF 16206123 research grant and X. Wang was supported by the RCF under Grant 359848.}}

\maketitle

\begin{abstract}
 Type-II sparse Bayesian learning (SBL) for the multiple measurement vector model recovers a sparse signal power vector by minimizing the marginal negative log-likelihood via the expectation-maximization (EM) algorithm, referred to as M-SBL.  We revisit the covariance learning (CL) problem and develop three parallel first-order algorithms within  the successive convex approximation (SCA) framework. CL-SCA minimizes per-coordinate surrogates formed by partially linearizing the log-determinant term, while CL-FP uses a damped projected fixed-point iteration derived from stationarity conditions. Both methods have closed-form updates. We prove that the fixed points of both update maps are exactly the KKT points of the CL problem, that each update direction satisfies a quantitative descent inequality, that the damped iterates achieve sufficient decrease with an $\mathcal{O}(T^{-1/2})$ best-iterate stationarity rate, and we derive a backtracking rule that retains this rate without step-size tuning. If damping factors are bounded away from zero, the residual vanishes along the whole sequence and every limit point is a KKT point. Local analysis motivates CL-Hyb, which combines CL-SCA and CL-FP. Experiments on device activity detection in massive random access show that CL-Hyb-bt matches the best CL solvers' accuracy at a fraction of the runtime and is over an order of magnitude faster than EM.
\end{abstract}

\begin{IEEEkeywords}
Activity detection, covariance learning,  grant-free random access,  multiple measurement vector model, sparse Bayesian learning, sparse signal recovery, successive convex approximation, Type-II likelihood.
\end{IEEEkeywords}
\IEEEpeerreviewmaketitle

 \definecolor{crimson2143940}{RGB}{214,39,40}
\definecolor{darkgray176}{RGB}{176,176,176}
\definecolor{darkorange25512714}{RGB}{255,127,14}
\definecolor{forestgreen4416044}{RGB}{44,160,44}
\definecolor{mediumpurple148103189}{RGB}{148,103,189}
\definecolor{steelblue31119180}{RGB}{31,119,180}
 
\tikzset{
  hyb/.style ={semithick, steelblue31119180,     solid,   mark=*,        mark size=1.2},
  cwo/.style ={semithick, darkorange25512714,    dashed,  mark=square,   mark size=2.2,
               mark options={solid, fill=none}},
  clmm/.style={semithick, forestgreen4416044,    dotted,  mark=o,        mark size=3.3,
               mark options={solid, fill=none}},
  clmp/.style={thick,     crimson2143940,        solid,   mark=triangle*,mark size=2.2},
  msbl/.style={thick,     mediumpurple148103189, solid,   mark=diamond*, mark size=2.2},
}

\vspace{-0.3cm}
 
\section{Introduction}
\label{sec:intro}

\IEEEPARstart{I}{n the}  \emph{multiple measurements vector} (MMV) \cite{duarte2011structured} model, the measurement vectors  $\y_t \in \mathbb{C}^M$ follow the generative model $
\y_{t} = \A \x_{t}  + \mathbf{e}_{t}$, $t = 1,\ldots,  \ns$,
where $\A=(\a_1  \ \cdots \ \a_{N}) \in \C^{\pdim \times N}$ is a fixed (known) overcomplete matrix ($N > M $) 
and   $\mathbf{e}_t  \in \C^{\pdim}$ is an unobserved zero mean white noise random vector,  i.e., $\cov(\e_t)=\sigmaw \mathbf{I}$. 
Matrix $\A$ is  called the \emph{dictionary} or \emph{measurement matrix}, and its column vectors $\mathbf{a}_i$ are referred to as \emph{atoms}. 
The unobserved random  signal vectors $\x_t = (x_{1t}, \ldots,x_{Nt})^\top$, $t=1,\ldots,\ns$ are assumed to be jointly sparse, {\it i.e.,} most of their elements are zero, and all $\x_t$-s share the same support.   
 Letting $\mathbf{Y} = ( \y_1 \, \cdots \,  \y_{\ns} ) \in \C^{\pdim \times \ns}$ denote the matrix of measurement vectors, we can write the model in matrix form as  
 \beq \label{eq:MMVmodel} 
\mathbf{Y}  = \A \mathbf{X} + \mathbf{E},
\eeq 
where matrices $\X  = (x_{nt}) \in \C^{N \times \ns}$ and  $\mE  = (e_{mt})  \in \C^{M \times \ns} $  contain the  signal  and noise vectors as columns, respectively. 

Given $K$ signals, the key assumption underlying the  MMV model  is that the signal matrix $\mathbf{X}$ is $\kdim$-rowsparse, {\it i.e.,} at most  $\kdim$ rows of $\X$ contain
non-zero entries. The rowsupport  of  $\X \in \C^{N \times \ns}$ is the index set of rows  containing non-zero  elements:
\[
\mathcal M = \rsupp(\X) 
= \{   n \in  [\! [   N ] \!]   :  x_{nt} \neq 0 \,  \mbox{ for some $t\in [\! [ \ns ] \!] $} \} 
\]
where $[\! [  N  ] \!]  = \{1,\ldots, N \}$. 
Since $\X$ is $\kdim$-rowsparse,  i.e.,  $| \rsupp(\X)| = \kdim$, 
joint estimation of $\x_t$-s using \eqref{eq:MMVmodel} can lead  both to computational advantages and increased reconstruction accuracy \cite{tropp2006algorithms,tropp:2006,chen_huo:2006,eldar_rauhut:2010,duarte2011structured,blanchard2014greedy}. 
 In sparse signal recovery (SSR) problems, the goal is identifying the support $\mathcal M$ and possibly $\X$, given only the data $\Y$, the dictionary $\A$, and possibly the sparsity level $K$. 
 
In Sparse Bayesian Learning (SBL)~\cite{tipping2001sparse,wipf2004sparse}, the signal vectors $\mathbf{x}_t$-s  are modeled as  i.i.d. circular Gaussian random vectors with independent elements and zero mean.  Additionally, $\mathbf{x}_t$-s are assumed to be statistically independent of the noise vectors $\e_t$-s. The noise vectors are assumed to be i.i.d. Gaussian with $\e_t \sim \mathcal{CN}_M(\mathbf{0}, \sigma^2 \mathbf{I})$. This implies that  
$ \y_t \sim \mathcal{C N}_{\pdim}(\mathbf{0}, \M)$, where the positive definite Hermitian (PDH) $\pdim \times \pdim$ covariance matrix $\M=\cov(\y_t)$ has the form 
\beq \label{eq:M}
\M = \A \Gam \A^\hop + \sigma^2 \mathbf{I}  =  \sum_{n=1}^N \gamma_n \a_n \a_n^\hop + \sigma^2 \mathbf{I}  
\eeq
where the signal covariance matrix $\Gam = \cov(\x_t) = \diag(\gam)$  is a diagonal matrix and $\gam \in \mathbb{R}^N_{\geq 0}$ is a vector of signal powers with only $K$ non-zero elements, i.e.,  $\mathcal M = \rsupp(\X)=\supp(\gam)$.    We assume $\sigma^2>0$ to be known. In many applications, the noise variance $\sigma^2$ can be accurately estimated a priori and can thus be considered as a known constant. 

\emph{Covariance learning} (CL-)based support recovery algorithms can be constructed by  
minimizing the negative log-likelihood function (LLF) of the data $\Y$. 
 We let  
$$\phi(\cdot ; \bmu, \M) =  | \M |^{-1}  \exp(- [\cdot-\bmu]^\hop \M^{-1} [\cdot-\bmu])$$ 
denote the probability density function  (up to an irrelevant normalizing constant) of the multivariate circular Gaussian distribution $\mathcal{CN}_M(\bmu, \M)$ with mean  vector $\bmu$ and covariance matrix $\M$. Here,  $| \cdot |$ denotes the matrix determinant. 

Then, the scaled negative LLF of the data $\Y $ is 
\beq
\label{eq:conv+nonconv}
\ell(\gam) = -\frac{1}{\ns} \log \Big(\prod_{t=1}^\ns \phi(\y_t ; \mathbf{0}, \M) \Big) = \gfnc(\gam) +  \ffnc(\gam),
\eeq
where $g : \mathbb{R}^N_{\geq 0} \to \R$ is a convex continuously differentiable in $\mathbb{R}^N_{\geq 0}$ 
and $f : \mathbb{R}^N_{\geq 0} \to \R$ is a concave continuously differentiable function, 
 defined as, 
\begin{align}
\gfnc(\gam) &= \Tr ( (\bm{A}\diag(\gam)\A^\hop + \sigma^2\,\bm{I})^{-1}\,\S ) \label{eq:def:conv}\\
\ffnc(\gam) &= \log |\bm{A}\diag(\gam)\bm{A}^\Hmat + \sigma^2\,\bm{I} | ,\label{eq:def:nonconv}
\end{align}
where $\tr(\cdot)$ denotes the matrix trace, and 
$ \SCM = \frac{1}{\ns} \sum_{t=1}^{\ns} \y_{t} \y_t^\hop$ 
is the sample covariance matrix (SCM) of the data.  We then study the nonconvex optimization problem
 \begin{equation}\label{eq:orig_formulation}
\begin{array}{ll}
\underset{\gam \geq \mathbf{0}}{\text{minimize}} & \ell( \boldsymbol{\gamma})
\end{array}
 \end{equation}
and derive  efficient algorithms using the SCA framework \cite{scutari2013decomposition}. 

\subsection{Relevant Prior Work}

A wide range of methods have been developed for the MMV problem in~\eqref{eq:MMVmodel}. Convex relaxation methods treat $\X$ as deterministic and promote row sparsity through $\ell_{2,1}$-type penalties~\cite{tropp:2006}. Greedy methods, such as simultaneous orthogonal matching pursuit (SOMP)~\cite{tropp2006algorithms}, recover the common support sequentially at low computational complexity.
Another class of deterministic row-sparse recovery methods relies on iterative reweighting or non-convex sparsity-promoting penalties, with M-FOCUSS~\cite{cotter2005sparse} as a representative example. When the source and noise vectors are modeled as independent random Gaussian vectors, the resulting marginal likelihood leads to Type-II objective in \eqref{eq:conv+nonconv}  as in SBL \cite{tipping2001sparse,wipf2004sparse,wipf2007empirical}. 

SBL is the most prominent Type-II method, introduced for the single measurement vector (SMV) setting  in~\cite{tipping2001sparse,wipf2004sparse} and extended to MMV as
M-SBL~\cite{wipf2007empirical}. The traditional expectation-maximization (EM) solver for M-SBL is known to converge slowly, especially in large-scale problems. In addition, with strictly
positive initialization, EM's multiplicative updates do not produce exact zeros. These drawbacks have motivated faster or more flexible Type-II variants, including MM-based covariance learning algorithm (CLMM) \cite{esfandiari2026mm}, generalized approximate message passing (GAMP)-based implementations with linear per-iteration complexity \cite{al2017gamp}, variational methods for jointly learning signal powers and full-structure noise covariance matrices~\cite{hashemi2022joint}. 

Greedy methods are attractive for their simplicity and low complexity. The greedy principle has recently been adopted in the covariance domain:  covariance learning matching pursuit (CL-MP)~\cite{marata2025activity}, independently developed as light-weight sequential SBL (LWS-SBL)~\cite{pote2025theory}, sequentially selects atoms by optimizing a Type-II (sequential) likelihood and stops after $K$ iterations. Although effective at low sparsity levels, greedy methods can degrade at higher
sparsity levels due to accumulated support-selection errors.

Covariance learning has also emerged as a powerful tool for activity detection (AD) in grant-free random access (GFRA), a key massive machine-type communication (mMTC) application~\cite{liu2018sparse}. Due to sporadic traffic, the received pilots follow an MMV model with a row-sparse channel matrix, making AD a support-recovery problem. Although many works use compressive sensing (CS) formulations with AMP-type solvers \cite{liu2018massive,chen2018sparse}, Type-II/CL methods can detect more active devices with the same pilot length when sufficiently many base station (BS) antennas are available~\cite{fengler2021non}. Here the CL variables represent device powers, namely activity indicators scaled by large-scale fading and transmit powers. Since EM-based M-SBL, coordinate-wise optimization (CWO)~\cite{fengler2021non}, and CL-MP/LWS-SBL \cite{marata2025activity,pote2025theory} optimize the same Type-II likelihood, they provide natural benchmarks for the proposed algorithms. Further related developments include joint activity and data detection \cite{chen2019covariance}, phase-transition analysis~\cite{chen2021phase}, robust CL-based AD~\cite{wang2025robust}, and SCA-based ML/MAP detectors \cite{tan2025fast}. 

\subsection{Contributions}

Our contributions are summarized as follows.

1) We propose a parallel SCA algorithm (CL-SCA) for covariance learning that preserves the convex part of the surrogate objective and linearizes only the concave part. This yields $N$ convex per-coordinate subproblems whose exact minimizers are  given in closed-form at each iteration (Theorem~\ref{th:SCA}) and can be updated in parallel. The per-iteration computational complexity is $\mathcal{O}(NM^2)$, matching one cycle of CWO \cite[Alg.~1]{fengler2021non}, and one iteration of CL-MP \cite[Alg.~1]{marata2025activity}, M-SBL \cite{wipf2007empirical} or CLMM \cite{esfandiari2026mm}.  

2) We derive a fixed-point form of the KKT conditions (Lemma~\ref{lem:fixed_point}) and derive CL-FP algorithm as a damped fixed-point iteration (obtained earlier as PSCA-ML-UD  \cite{tan2025fast} from a different surrogate) and give the first rate and step-size analysis. Its per-iteration cost equals that of CL-SCA.

3) We provide a rigorous convergence analysis showing that the fixed points of both the CL-SCA and CL-FP update maps coincide exactly with the KKT points of the CL problem. In addition, each update direction satisfies a quantitative descent inequality, and the damped iterates achieve sufficient decrease and an $\mathcal{O}(T^{-1/2})$ best-iterate stationarity rate after $T$ iterations. Moreover, when the damping factors are bounded away from zero, the residual vanishes along the entire sequence, and every limit point is a KKT point ({\it cf.} Theorems~\ref{th:g_rate}~and~\ref{th:fp_rate}). 

4) We then propose a hybrid algorithm (CL-Hyb) exploiting the benefits of both mappings. Local analysis shows that CL-SCA allows twice the step size of CL-FP, whereas CL-FP is faster near stationarity. This motivates CL-Hyb, which uses CL-SCA for
the first few iterations and then switches to CL-FP. Since both updates satisfy the same sufficient-decrease inequality with the same constants, the convergence guarantees carry over to the hybrid scheme. 

5) We provide a unified fixed-point view of CL-SCA, CL-FP, CLMM, and EM as updates $\gamma_n\mapsto\gamma_n+(\text{multiplier})\times(\text{residual})$
(Table~\ref{tab:update_family}), showing that the fixed points of the first three are exactly the KKT points, whereas EM's vanishing multiplier at $\gamma_n=0$ creates spurious fixed points on the boundary.

6) We apply the proposed methods to joint activity detection and channel
estimation (JADCE) in GFRA and compare them with M-SBL~\cite{wipf2007empirical},
CLMM~\cite{esfandiari2026mm}, CWO~\cite{fengler2021non}, and CL-MP~\cite{marata2025activity}. CL-Hyb with backtracking (CL-Hyb-bt) attains the missed-detection performance of CWO at roughly half its runtime  and is more than an order of magnitude faster than M-SBL (cf.\ \autoref{fig:CPU}). The objective-gap analysis in
\autoref{fig:loss_curve} further shows that CL-Hyb-bt converges fastest among the optimization-based methods, in both iterations and runtime.

For reproducible research, code for all proposed and benchmark methods is available at \url{https://github.com/esollila/CL-SCA}. Preliminary results appeared
in~\cite{ollila2026joint,esfandiari2026mm}.

\subsection{Organization and Notation}

 The paper is structured as follows.  \autoref{sec:1st_order} gives the 1st-order and fixed-point characterizations of the CL problem and relates them to CWO. \autoref{sec:MSBL}  reviews M-SBL and its EM solver. Section~\ref{CL-SCA-AL_SEC} develops the proposed CL-SCA algorithm and its stationarity-rate analysis, and discusses step-size selection, including a backtracking rule that attains the $\mathcal O(T^{-1/2})$ rate without step-size tuning. Section~\ref{sec:CL-FP} develops CL-FP and its rate analysis, and casts CL-SCA, CL-FP, CLMM, and EM in a unified fixed-point view in \autoref{subsec:mm_vs_sca}. The local convergence analysis in \autoref{subsec:convergence} then motivates the third algorithm, CL-Hyb. Section~\ref{sec:source}  applies the methods to JADCE in GFRA, and Section~\ref{sec:concl}  concludes the paper.

We use bold lowercase and uppercase letters for vectors and matrices,  respectively. For a matrix $\X, \X^\hop$, $\tr(\X)$, $|\X|$, and $\|\X\|_{2}$ denote its Hermitian transpose, trace, determinant, and
spectral norm. The entry in the $m^{\rm th}$ row and $n^{\rm th}$ column of $\X$ is denoted by $[\X]_{mn}$, and its largest eigenvalue is denoted by $\lambda_{\max}(\X)$. The Euclidean norm of a vector $\x$ is denoted by $\|\x\|_{2}$, and $[\x]_{+}$ denotes the element-wise projection of $\x$ onto the nonnegative orthant. The symbol $\odot$ denotes the Hadamard (element-wise) product.  $\mathsf{1}\{\cdot\}$ denotes the indicator function.

\section{First order characterization of the problem} \label{sec:1st_order}

Write 
\begin{align} 
\M_{\backslash n} &= \sum_{j \neq n}  \gamma_j   \a_j \a_j^\hop + \sigmaw \mathbf{I}  = \M - \gamma_n \a_n \a_n^{\hop} \label{eq:Mbackslash}  
\end{align}
as the covariance matrix of  $\y_t$-s when the contribution from the $n^{\text{th}}$ signal is removed. Define  
\beq \label{eq:b}
\tildeb_n =  \M^{-1}_{\backslash n}  \a_n \qquad \mbox{and}  \qquad  \b_n =  \M^{-1}  \a_n, 
\eeq  
for $n=1,\ldots, N$. Note that $\tildeb_n\equiv\tildeb_n(\gam)$ and $\b_n\equiv\b_n(\gam)$  depend on $\gam$ through $\M$ but to keep notations short,  we have dropped this notation.  
Write $\nabla_n = \frac{\partial}{\partial \gamma_n}$ as a shorthand notation. Then   the coordinate derivatives $ d_n(\gam)=\nabla_{n}\ffnc(\gam)$ and $e_n(\gam) =  - \nabla_{n} \gfnc(\gam)$ of functions $f(\gam)$ and $g(\gam)$ are 
\begin{align}
 d_n(\gam) 
&= \nabla_n \log|\M |
  =   \a_n^\hop   \M^{-1} \a_n  
  = \a_n^\hop \b_n  \label{eq:partial_gammai_ncvx}  \\
e_n(\gam) 
&=  -\!   \nabla_n \tr(\M^{-1} \S) \!=\! \a_n^\hop \M^{-1}\mathbf{S} \M^{-1}\a_n 
 =\b_n^\hop  \mathbf{S} \b_n
   \label{eq:partial_gammai_ncvx2}
\end{align} 
and note that $d_n(\gam)>0$ as $\M \succ 0$ and $e_n(\gam) \geq 0$ for all $\gam$. For notational simplicity we often write $d_n\equiv d_n(\gam)$ and $e_n\equiv e_n(\gam)$, dropping the dependency of $\gam$.  
With the above definitions,  the $n^{\text{th}}$ coordinate function of the gradient of $\ell(\gam)=f(\gam) + g(\gam)$ can be expressed as:
\beq \label{eq:grad_ell}
\nabla_n \ell(\gam)
= \a_n^\hop \b_n -  \b_n^\hop \S \b_n = d_n- e_n
\eeq
The KKT condition for the constrained case $\gam \geq \mathbf{0}$ is  
\beq \label{eq:KKT}
\nabla_n \ell(\gam)  \begin{cases} = 0, & \mbox{if $\gamma_n >0$} \\
\geq 0, & \mbox{if $\gamma_n = 0$}  \end{cases}.
\eeq   
Next result characterizes the solution to \eqref{eq:orig_formulation}. 

\begin{lemma} \label{lem:fixed_point}
Let $\gam\in\mathbb{R}^N_{\ge0}$ and let $H(\gam)$ have coordinate
functions
 \begin{align}
 H_n(\gam)
&= \gamma_n +  \frac{1}{d_n}\bigg(\frac{e_n }{d_n }  - 1 \bigg)   \label{eq:gamma_i_sol2} \\
&= \gamma_n-\frac{\nabla_n\ell(\gam)}{d_n^2}. \label{eq:H_scaled_gradient}
\end{align}
Then $\gam$ satisfies the KKT conditions \eqref{eq:KKT} if and only if it verifies the
fixed-point equation
\beq \label{eq:lik_ver2}
\gam =  [H(\gam)]_+ .
\eeq
In particular, every local minimizer of $\ell$ over
$\mathbb{R}^N_{\ge0}$, and hence the minimizer
$\widehat{\gam}$, verifies \eqref{eq:lik_ver2}, and on the support
$\{n:\gamma_n>0\}$ the projection is inactive, so
$\gamma_n=H_n(\gam)$.
\end{lemma}

\begin{proof} See Appendix~\ref{app:lem:fixed_point}.
\end{proof}

The  CWO algorithm is closely related with fixed point map. Namely, in cycle $k$, coordinate $n$ is updated by minimizing
\begin{align*} 
\ell_{\text{cw},n}(\gamma)= \ell( \gamma_1^{k}, \ldots,  \gamma_{n-1}^{k}, \gamma, \gamma_{n+1}^{k-1}, \ldots, \gamma_N^{k-1})
\end{align*}
subject to $\gamma \geq 0$. The minimizer is   \cite{yardibi2010source}: 
\begin{align} \label{eq:gamma_i_star}
\gamma_n^k  \!&=\!  \arg \min_{\gamma \geq 0}  \ \ell_{\text{cw},n}(\gamma)
= [H_n(\tilde{\gam}^{\iter})]_+
\end{align} 
where $\tilde{\gam}^{\iter} = \left(\gamma_1^{k}, \ldots,  \gamma_{n-1}^{k}, \gamma_n^{k-1} , \gamma_{n+1}^{k-1},  \ldots,\gamma_N^{k-1} \right)^\top.$

\section{M-SBL} \label{sec:MSBL}

 Under the Gaussianity assumption, the posterior distribution of the sources factorizes across measurement vectors as
$$
p(\X \mid \Y, \gam)
= \prod_{t=1}^{Q} \phi\!\big(\x_t;\, \bmu_t, \M_{x}\big).$$
The posterior depends on   the conditional mean vectors  $\bmu_t \in \mathbb{C}^N$ of $\x_t$-s as well as  their joint posterior covariance matrix, $\M_x \in \mathbb{C}^{N \times N}$, defined by 
\begin{align*} 
&\mathbf{M}(\gam)=\E[\X \mid \Y, \gam]
= (\bmu_1 \, \cdots \, \bmu_Q)  
= \Gam \A^\hop \M^{-1} \Y, \\ 
&\M_x(\gam) = \cov(\x_t \mid \y_t, \gam)=  \Gam - \Gam \A^\hop \M^{-1} \A \Gam,
\end{align*} 
where $\M$ is defined by \eqref{eq:M}.
The M\text{-}SBL algorithm~\cite{wipf2007empirical,wipf2007beamforming} employs an empirical Bayes (Type-II) approach to estimate the posterior distribution. In the Type-II framework, rather than placing a prior on the hyperparameters~$\gam$, one estimates them directly from the data by minimizing the marginal Type-II scaled negative LLF in \eqref{eq:conv+nonconv}, thereby obtaining $\hat{\gam}$. Subsequently, the posterior distribution is computed via the plug-in principle as \(p(\X \mid \Y, \hat{\gam})\), and
\begin{align*}
\hat{\X} &= \mathbf{M}(\hat \gam)= \hat{\Gam} \A^\hop \hat{\M}^{-1} \Y, 
\\
\hat{\M}_x &= \hat{\M}_x(\hat \gam)=\hat{\Gam} - \hat{\Gam}\,\A^\hop \hat{\M}^{-1} \A \hat{\Gam},
\end{align*}
will serve as the empirical Bayes estimates of the posterior mean vectors and covariance matrix, respectively, where  $\hat{\Gam} = \diag(\hat{\gam})$ and $\hat{\M} = \A \hat{\Gam} \A^{\hop} + \sigma^2 \mathbf{I}$.  Algorithm~\ref{alg:SBL-JADCE} summarizes the \emph{generic} M-SBL framework: step~1, solved by EM in \cite{wipf2007empirical}, can use any solver of \eqref{eq:orig_formulation}, such as the proposed SCA algorithms.

\subsection{The EM-algorithm}

In step~1 of Algorithm~\ref{alg:SBL-JADCE}, the M-SBL constructs an EM algorithm to minimize $\ell(\gam)$ by treating $(\y_t,\x_t)$, $t=1,\ldots,\ns$, as the complete data.  It is well-known that the EM algorithm suffers from slow convergence, making M-SBL  impractical for applications where $N$ is large. 

\begin{algorithm}[!t]
\caption{Generic  M\text{-}SBL framework}\label{alg:SBL-JADCE}
\DontPrintSemicolon
\SetKwInOut{Input}{Input} 
\SetKwInOut{Output}{Output}
\SetKwInOut{Init}{Initialize}
\SetNlSkip{1em}
\SetInd{0.5em}{0.5em}
\Input{$\mathbf{Y}$, $\A$, $\sigma^2$, $K$ or $\gamma_{\text{thr}}$.}

Minimize $\ell(\gam)$ over $\gam \geq 0$ to obtain $\hat \gam \in \mathbb{R}^N_{\geq 0}$ using an algorithm of your choice. 

Compute estimate of activity indices $\hat{\boldsymbol{\alpha}}\! = \!(\hat\alpha_1,\ldots,\hat\alpha_N)^{\top} \!\in \! \{0,1\}^N$ using either the threshold rule in  \eqref{eq:threshold_rule} or the $K$-top rule in \eqref{eq:topKrule}.   

Prune the signal power estimates as  $\hat{\gam}  \gets \hat{\gam} \odot \hat{\boldsymbol{\alpha}}$.  

Compute the sparse signal matrix as  
$
\hat{\X} = \hat{\Gam} \A^\hop \hat{\M}^{-1} \Y, 
$
with $\hat{\Gam} = \diag(\hat{\gam})$ and $\hat{\M} = \A \hat{\Gam} \A^{\hop} + \sigma^2 \mathbf{I}$. 

 \Output{$\hat{\boldsymbol{\alpha}}$ (or support $\hat{\mathcal M} = \supp(\hat{\boldsymbol{\alpha}})$), $\hat \X$, $\hat{\gam}$.} 
\end{algorithm}

To highlight the analogy with the SCA algorithms developed in this work, we combine the E- and M-steps of M-SBL into a concise efficient form that avoids computing the posterior mean and covariance estimates at each iteration. Specifically, 
let $\gam^k$ denote the value of $\gam$ at $k^{\text{th}}$ iteration, and write 
\beq \label{eq:dn_and_en}
d_n^k=d_n(\gam^k)= \a_n^\hop \b_n^k \ \mbox{and} \ e_n^k = e_n(\gam^k) = [\b_n^k]^\hop \S \b_n^k
\eeq
i.e., as values the coordinate derivative functions defined in \eqref{eq:partial_gammai_ncvx} and \eqref{eq:partial_gammai_ncvx2} evaluated at $\gam = \gam^k$,  respectively. Also, write $\mathbf{M}^k = \mathbf{M}(\gam^k)$ and $\M_x^k=\M_x(\gam^k)$ for the values of the posterior mean and covariance matrix at $\gam=\gam^k$.  With the above definitions in mind, 
we can write      
$[\M_x^k]_{nn} = \gamma_n^k  - (\gamma_n^k)^2 d_n^k $ and 
$\big[ Q^{-1}   \mathbf{M}^{\iter} ( \mathbf{M}^{\iter})^\hop \big]_{nn}   = (\gamma_n^k)^2e_n^k$, and the M-step \cite[eq. (18)]{wipf2007empirical} can be expressed in a simpler form as 
\beq \label{eq:MSBL_update}
\gamma_n^{k+1} = \gamma_n^k  - ( \gamma_n^k)^2 ( d_n^k -  e_n^k )  
= \gamma_n^k(1  - \gamma_n^k( d_n^k -  e_n^k )  ). 
\eeq
Thus, the M-SBL algorithm consists of iterating \eqref{eq:MSBL_update} until convergence. Note that due to multiplicative form of the update, one needs to start from some strictly positive initial point (e.g., $\gam^0 = \mathbf{1}$ is proposed in \cite{wipf2007empirical}).  Since $\nabla_n \ell(\gam^k)  
=d_n^k - e_n^k$ by \eqref{eq:grad_ell},  we can interpret \eqref{eq:MSBL_update}  as a gradient descent method with coordinate-wise adaptive step size $(\gamma_n^k)^2$.

Apart from its slow convergence, another limitation of the EM iterations \eqref{eq:MSBL_update} is that they never produce an exactly sparse solution. Indeed, the EM minimization requires strictly positive hyperparameters, since the objective function in the E-step is defined strictly in the domain $\gam >0$  which matches the natural parameter space of the underlying probabilistic model, so the solution in M-step is strictly positive. If $\gamma_n^0>0$ (i.e., $n^{\text{th}}$ component of the initial start is strictly positive), then the M-step yields $\gamma_n^{k+1}>0$ for all $k$. Conversely, if $\gamma_n^0=0$, then $\gamma_n^k=0$ for all $k$, i.e., once a component is initialized at zero, it remains zero throughout the EM iterations. Consequently, under positive initialization, M-SBL returns iterates that are dense for every finite $k$, although its components may become arbitrarily small as the iterations proceed.

\subsection{Sparsifying the solution} 

The (local) minimizer $\hat \gam$ of $\ell(\gam)$ found by any of the CL methods (CWO, EM of M-SBL, or the proposed CL-SCA, CL-FP, or CL-Hyb) may not be sufficiently sparse.  For this reason, in Step~3 of Algorithm~\ref{alg:SBL-JADCE}, one prunes small signal powers to zero via $\hat{\gam}  \gets \hat{\gam} \odot \hat{\boldsymbol{\alpha}}$, where the activity indices $\hat{\boldsymbol{\alpha}}$ are obtained using either of the following rules: 
\begin{itemize}
    \item The \emph{threshold-rule}:  
\beq \label{eq:threshold_rule}
\hat{\alpha}_n = \mathsf{1}{\{\hat \gamma_n \geq \gamma_{\text{thr}} \}} = \begin{cases} 1, & \text{if $\hat \gamma_n \geq \gamma_{\text{thr}}$ } \\ 0 & \text{otherwise} \end{cases} 
\eeq
where $\gamma_{\text{thr}}$ is a fixed threshold.  
\item  The \emph{top $K$-rule}:
\beq \label{eq:topKrule}
\hat \alpha_n = \mathsf{1}\!\left\{ \gamma_n \text{ is among the $K$ largest entries of } \hat{\gam} \right\}
\eeq
where the sparsity level $K$ is assumed to be known.
\end{itemize}
The latter requires knowledge (or an estimate) of  $K= | \mathcal M|$ while the former requires calibrating a suitable threshold, either using a priori knowledge or via estimation. The threshold can be designed such that a desired trade-off between the probabilities of missed detection and false alarm is achieved.

\section{CL-SCA algorithm}  \label{CL-SCA-AL_SEC}

Under the  SCA framework, each agent in the optimization focuses on a single variable, $\gamma_1, \gamma_2, \dots, \gamma_N$, and the resulting updates are performed iteratively and in parallel.

\subsection{CL-SCA algorithm} \label{3_1}

At iteration $k+1$, by following the SCA framework and linearizing the concave term $f(\gam) = \log |\M|$, the $N$ functions  to be minimized by the $N$ agents in parallel  are given as
\beq \label{eq:tilde_f}
\begin{aligned}
\ell_n(\gamma_n \mid \gam^{\iter}) 
= \gfnc\big(\gamma_1^k, \ldots, \gamma_{n-1}^k,\gamma_n, \gamma_{n+1}^k, \ldots, \gamma_N^k \big) 
\\ \qquad + d_n^k \,(\gamma_n - \gamma_n^{\iter}), \ n \in \{ 1, \ldots, N\}, 
\end{aligned}
\eeq 
where $d_n^k$ is defined in \eqref{eq:dn_and_en}. 
The solution is given below. 

\begin{theorem}  \label{th:SCA} For $n \in \{1,\ldots, N\}$, the minimizer of the  convex function $\ell_n(\gamma_n \mid \gam^{\iter})$  over $\gamma_n \geq 0$ is 
\begin{align}  
\hat{\gamma}_n(\gam^{k}) 
&= \left[ 
\gamma_n^{\iter} + \frac{1}{d_n^k} \Bigg( \sqrt{\frac{e_n^k}{d_n^k}} - 1 \Bigg) 
\right]_+ \label{eq:SCA_optim_gammai} 
\end{align} 
 where $d_n^k$ and $e_n^k$ are defined in \eqref{eq:dn_and_en}.
\end{theorem} 

\begin{proof} See Appendix~\ref{app:th:SCA}.
\end{proof} 

The update \eqref{eq:SCA_optim_gammai} coincides with the one derived in the conference version \cite[Th.~1]{ollila2026joint} of this work, 
but it is also closely related to the majorization-minimization (MM)-based covariance learning algorithm (CLMM) of \cite{esfandiari2026mm}: the two were obtained independently and look dissimilar, but they move each coordinate in the same direction and differ only in one scalar multiplier.  We defer this discussion to subsection~\ref{subsec:mm_vs_sca}.

Once  solutions in \eqref{eq:SCA_optim_gammai}  are computed for each signal power  in parallel,  denoted $\hat \gam(\gam^k) = (\hat \gamma_1(\gam^k), \ldots, \hat \gamma_N(\gam^k))^\top$, the next iterates of the algorithm are formed by including a smoothing (damping) operation with the previous ones, which is necessary to
guarantee convergence, as
\begin{align}
\gam^{\iter+1} &= \big(1 - \eta^k \big)\gam^{\iter} + \eta^k\,\hat{\gam}\left(\gam^k\right) 
=  \gam^{\iter} + \eta^k\, \mathbf{p}^k  \label{eq:SCA_smoothing2} 
\end{align}
where the step-size $\eta^k \in (0,1]$ is introduced to control the length of the update along
the update direction 
\beq \label{eq:g_direction}
\mathbf{p}^k = \hat{\gam}(\gam^k) - \gam^k.
\eeq
The step-size sequence $\{ \eta^k \}$   can be conveniently chosen as a diminishing sequence, verifying 
 $\sum_{k=0}^{\infty} \eta^k = \infty$ and $\sum_{k=0}^\infty (\eta^k)^2 < + \infty$ \cite{scutari2013decomposition,scutari2018parallel}. 
 However, our convergence analysis suggests a backtracking algorithm that certifies the best-iterate stationary rate. 
 We use $\mathbf{\gam}^0 = \mathbf{0}$ as default initial start. The CL-SCA algorithm is tabulated in \autoref{alg:SCA}.   

\begin{algorithm}[t]
 \caption{\textsf{CL-SCA} algorithm}\label{alg:SCA}
\DontPrintSemicolon
\SetKwInOut{Input}{Input} 
\SetKwInOut{Output}{Output}
\SetKwInOut{Init}{Initialize}
\SetNlSkip{1em}
\SetInd{0.5em}{0.5em}
\Input{$\mathbf{S}$, $\A$, $\sigma^2$, step-size sequence $\{\eta^k\}$, $I_{\max}$.} 
\Init{$\gam^0=\mathbf{0}$, $\delta=10^{-3}$, $k=0$.}

\nonl \Repeat{ $\bigl\|\gam^k -[\gam^k -  \nabla \ell(\gam^k)]_+\bigr\|_\infty < \delta$ \KwSty{or}  $ k \ge I_{\max}$}{ 
   
 \BlankLine
Compute $\M^k= \A \diag(\gam^k) \A^{\hop} + \sigma^2 \mathbf{I}$   
and $\B^k= ( \M^k )^{-1} \A = (\b_1^{\iter} \, \cdots \,\b_N^{\iter})$. 
\label{alg:cov_update_step}

For all $n \in \{1,\ldots, N\}$ solve in parallel: 
\[
\hat{\gamma_n}(\gam^{k}) = \left[ 
\gamma_n^{\iter} + \frac{1}{d_n^k} \Bigg( \sqrt{\frac{e_n^k}{d_n^k}} - 1 \Bigg) 
\right]_+
\]
where $d_n^k= \a_n^\hop \b_n^k$ and $e_n^k = [\b_n^k]^\hop \S \b_n^k$.

 \BlankLine
$ 
 \gam^{\iter+1} =\gam^{\iter} + \eta^k\,\mathbf{p}^k$, where $\mathbf{p}^k = \hat{\gam}(\gam^k) - \gam^k$.

$k = k+1$
 
 }
 \Output{$\hat{\boldsymbol{\gamma}}=\gam^{\iter}$.}
\end{algorithm}

\begin{lemma}
\label{lem:g_first_order}
A vector  $\gam\in\mathbb{R}_{\ge0}^N$,
 satisfies the KKT
conditions~\eqref{eq:KKT} if and only if it is a fixed point of
the update map of Theorem~\ref{th:SCA}, i.e.,
$\hat\gam(\gam)=\gam$.
\end{lemma}

\begin{proof}
See Appendix~\ref{app:lem:g_first_order}.
\end{proof}

Lemma~\ref{lem:fixed_point} and Lemma~\ref{lem:g_first_order} thus provide \emph{two fixed-point characterizations of stationarity}. Lemma~\ref{lem:fixed_point} and the SCA update map
$\hat\gam(\cdot)$ of Theorem~\ref{th:SCA} characterize the same set:
\begin{equation*}
    \hat\gam(\gam)=\gam
    \ \Leftrightarrow\
    \gam=[H(\gam)]_+
    \ \Leftrightarrow\
    \text{$\gam$ satisfies \eqref{eq:KKT}.} 
\end{equation*}
This is no coincidence: recalling $d_n=\a_n^\hop\b_n$, $e_n=\b_n^\hop\SCM\b_n$, and writting $r_n=e_n/d_n$, the unprojected
displacements of the two maps are
\begin{align}
    \hat\gamma_n(\gam)-\gamma_n&=\frac{\sqrt{r_n}-1}{d_n} \label{eq:SCA_optim_gammai_rn} \\
    H_n(\gam)-\gamma_n &=\frac{r_n-1}{d_n}
    =(1+\sqrt{r_n}\,)\,\frac{\sqrt{r_n}-1}{d_n},\label{eq:FP_optim_gammai_rn} 
\end{align}
thus they differ only by the positive factor $1+\sqrt{r_n}$. In other words, both maps move every coordinate in the same direction determined by the sign
of $r_n-1$, i.e., of $-\nabla_n \ell(\gam) = e_n-d_n$ and vanish  simultaneously. As termination criterion we use 
\begin{equation}    \label{eq:terminate_CL}
    \bigl\|\gam^k -[\gam^k -  \nabla \ell(\gam^k)]_+\bigr\|_\infty < \delta 
\end{equation}
where  $\nabla_n \ell(\gam^k) = d_n^k- e_n^k$ with  some appropriate  convergence threshold $\delta$  (default $\delta = 10^{-3}$). 
Note that \eqref{eq:terminate_CL} vanishes if and only if the KKT conditions~\eqref{eq:KKT} hold for $\gam^k$. In addition, note that \eqref{eq:terminate_CL} is an algorithm independent criterion\footnote{An alternative criterion is 
$\| \mathbf{p}^k \|_\infty \leq \delta$  which also certifies proximity to a KKT point of the original problem~\eqref{eq:orig_formulation}, but is algorithm specific.} 
that we will use also to terminate other competing iterative CL methods (CL-FP, CLMM, M-SBL, CWO). Computing $\nabla \ell(\gam^k)$ does not incur any extra costs, since $d_n^k$ and $e_n^k$ required by $\nabla_n \ell(\gam^k)$ are reused (not recomputed) in Step 2.  This is true for all methods considered. 

The computational complexity of step~1 is  $\mathcal{O}(N M^2) +    \mathcal{O}\!\left(M^{3} + M^{2} N\right)$ while the computational complexity of step 2 is 
$\mathcal{O}(NM^2)$.  Computational complexity of Step 3-4 are  $\mathcal{O}(N)$.   Since $N \geq M$, the complexity  of one iteration of CL-SCA algorithm is $\mathcal{O}(N M^2)$.    This is the same complexity per iteration as CWO \cite[Algorithm~1]{fengler2021non}, CL-MP \cite[Algorithm~1]{marata2025activity}, CLMM (cf. \autoref{subsec:mm_vs_sca}),  or the EM-algorithm \eqref{eq:MSBL_update} of M-SBL.   The input to CL-based approaches is  $\mathbf{S}$ (the SCM) with the computational complexity of $\mathcal{O}\!\left(Q M^2\right)$.  Thus, the overall complexities of CL-SCA, CLMM, EM algorithm of M-SBL and CWO are  $\mathcal{O}(N M^2 I_{\text{iters}} + Q M^2)$, where $I_{\text{iters}}$ is the number of iterations until convergence. Note that CL-MP is a greedy CL-based approach, terminated after $K$ iterations (when $K$ is known), and hence its overall complexity is  $\mathcal{O}(N M^2 K + QM^2)$. Thus CL-MP provides faster computations only in the ideal case when $K$ is known and $K<I_{\text{iters}}$. We provide convergence analysis of the CL-SCA (and CL-FP) algorithm in subsection~\ref{sec:convergence}. 

\subsection{Stationarity rate of CL-SCA}
\label{sec:convergence}

Define  $a_{\min}^2=\min_n\|\a_n\|^2$ and $a_{\max}^2=\max_n\|\a_n\|^2$  and note that both scalars are well-defined strictly positive constants as $\a_n\ne\bm 0$ for all $n$.   We say that $\nabla\ell$ is \emph{Lipschitz continuous} on a set $\mathcal C\subseteq\mathbb R^N_{\ge0}$ with constant $L_\nabla>0$ if 
\begin{equation}\label{eq:lipschitz_def}
  \|\nabla\ell(\gam)-\nabla\ell(\gam')\|_2\;\le\;L_\nabla\,\|\gam-\gam'\|_2
  \quad \forall\,\gam,\gam'\in\mathcal C .
\end{equation}
Since $\M(\gam)\succ\bm0$ on an open set containing $\mathbb R^N_{\ge0}$, $\ell$ is twice continuously differentiable there, and for convex $\mathcal C$ a sufficient condition for
\eqref{eq:lipschitz_def} is $\sup_{\gam\in\mathcal C}\|\nabla^2\ell(\gam)\|_2\le L_\nabla$ (see e.g., \cite[Lemma~1.2.2]{nesterov2004introductory}). Property~\eqref{eq:lipschitz_def} implies the \emph{quadratic upper bound} \cite[Prop. A.24]{bertsekas1999nonlinear}:
\begin{equation}\label{eq:quad_upperbound_L}
  \ell(\gam')\;\le\;\ell(\gam)+\nabla\ell(\gam)^{\top}(\gam'-\gam)
  +\frac{L_\nabla}{2}\|\gam'-\gam\|_2^2 \quad\forall \gam,\gam'\in\mathcal C ,
\end{equation}
which is used  later in \autoref{th:g_rate}. The next lemma collects the two constants that govern the step-size rule. 

\begin{lemma}\label{lem:g_constants}
Let  
 $$\mathcal K_0=\{\gam\in\mathbb{R}_{\ge0}^N: \ell(\gam)\le \ell(\gam^0)\}$$
denote the initial sublevel set. Then:
\begin{enumerate}
\item[{\rm (i)}]  $\mathcal K_0$ is compact, and 
$\lambda_{\max}(\M)\le\bar\lambda_0$ with $\bar\lambda_0=\sigma^{-2(M-1)}\exp\{\ell(\gam^0)\}$ 
for all $\gam\in\mathcal K_0$. 
Consequently
\begin{equation}    \label{eq:g_curvature_constant}
    \underline d    =\min_{\gam\in\mathcal K_0}\min_{1\le n\le N} d_n(\gam)    \ge   \frac{a_{\min}^2}{\bar\lambda_0}   > 0
\end{equation}
where $d_n(\gam)=\a_n^\hop \M^{-1} \a_n$. 

\item[{\rm (ii)}]  $\nabla \ell$ is Lipschitz continuous on all of $\mathbb{R}_{\ge0}^N$
with constant
\begin{equation}
    L_\nabla
    =
    \frac{a_{\max}^2\,\|\A\|_2^2}{\sigma^4}
    \left(1+\frac{2\|\SCM\|_2}{\sigma^2}\right).
    \label{eq:g_lipschitz_constant}
\end{equation}
\end{enumerate}
\end{lemma}

\begin{proof} See Appendix~\ref{app:lem:g_constants}.
\end{proof}

We can now state the sufficient decrease and stationarity rate.  Note that $\ell_{\inf}=\inf_{\gam\ge\bm0} \ell(\gam)>-\infty$ as $\ell(\gam) \geq M \log \sigma^2$.

\begin{theorem} \label{th:g_rate}
If $0<\eta^k\le\bar\eta$ for all $k$, where 
\begin{equation}
    \bar\eta=\min\{1,\underline d^{\,2}/L_\nabla\},
    \label{eq:g_safe_step}
\end{equation}
then the CL-SCA iterates~\eqref{eq:SCA_smoothing2} satisfy $\gam^k\in\mathcal K_0$ and
\begin{equation}  \label{eq:g_sufficient_decrease}
    \ell(\gam^{k+1}) \le \ell(\gam^k)-\frac{\underline d^{2}}{2} \eta^k \|\bm p^k\|_2^2   
\end{equation}
where $\underline d$ and $L_\nabla$ are as in Lemma~\ref{lem:g_constants}. Consequently, for every $T\ge1$,
\begin{equation}
    \min_{0\le k<T}\|\bm p^k\|_2^2
    \le
    \frac{2\{\ell(\gam^0)- \ell_{\inf}\}}
         {\underline d^{\,2}\sum_{k=0}^{T-1}\eta^k}
    \;\overset{(\gam^0=\bm 0)}{\le}\;
    \frac{2\,\tr(\SCM)/\sigma^2}
         {\underline d^{\,2}\sum_{k=0}^{T-1}\eta^k},
    \label{eq:g_rate}
\end{equation}
If, in addition, $\inf_k\eta^k>0$, then $\|\bm p^k\|_2\to0$ and every accumulation
point of $\{\gam^k\}$ is KKT-stationary.
\end{theorem}

\begin{proof}
See Appendix~\ref{app:th:g_rate}. 
\end{proof}

For a constant step size $\eta\in(0,\bar\eta]$, \eqref{eq:g_rate} yields the best-iterate stationarity rate $\min_{k<T}\|\bm p^k\|_2=\mathcal O(T^{-1/2})$, matching the canonical rate of first-order methods on smooth nonconvex problems. In contrast, diminishing schedules $\eta^k=c(k+1)^{-a}$ degrade the certified rate to $\mathcal O(T^{-(1-a)/2})$ for $a\in(1/2,1)$ and $\mathcal O((\log T)^{-1/2})$ for $a=1$, since $\sum_{k<T}\eta^k$ then grows only as $\Theta(T^{1-a})$ or $\Theta(\log T)$. Moreover, they do not meet the condition $\inf_k\eta^k>0$ under which Theorem~\ref{th:g_rate} guarantees that the residual vanishes along the whole sequence.  

\subsection{Step-size selection} \label{subsec:select_stepsize}

Based on \autoref{th:g_rate}, one could resort to a {\bf \emph{fixed stepsize} (fss)} $\eta=\bar\eta$ in~\eqref{eq:g_safe_step}, which maximizes the guaranteed decrease in~\eqref{eq:g_sufficient_decrease}. However, $\bar\eta$ depends on $\underline d$ and $L_\nabla$, which are not known a priori; computable bounds on them give an admissible but overly conservative step size.

Since $d_n^k=\a_n^\hop\b_n^k$ are computed in every iteration of Algorithm~\ref{alg:SCA}, a step size satisfying the conditions of Theorem~\ref{th:g_rate} can be found adaptively by the following {\bf \emph{backtracking rule}}: maintain the running minimum $\underline d_k=\min_{s\le k}\min_n d_n^s$, and halve a trial step $\eta^k\le 1$ until
$$
\ell(\gam^k+\eta^k\bm p^k) \le \ell(\gam^k)-(\underline d_k^{\,2}/2)\,\eta^k\|\bm p^k\|_2^2 .
$$
Since $\underline d_k\ge\underline d$, every accepted step satisfies the sufficient-decrease inequality~\eqref{eq:g_sufficient_decrease}, and the test is passed by every $\eta^k\le\bar\eta$. Hence the loop terminates after at most $\lceil\log_2(1/\bar\eta)\rceil$ halvings, the accepted steps
satisfy $\eta^k\ge\bar\eta/2$, and Theorem~\ref{th:g_rate} applies in full, including the $\mathcal O(T^{-1/2})$ rate and $\|\bm p^k\|_2\to 0$, without $\underline d$ or $L_\nabla$ having to be known.  Rather than restarting from $\eta=1$ in every iteration, the search starts from the last accepted step size, and doubles it after two consecutive iterations accepted at the first trial, so that the step can grow again. Every accepted step still satisfies $\eta^k\ge\bar\eta/2$, so the rate guarantee is unaffected. The rule is also inexpensive: evaluating $\ell$ at a trial point requires the Cholesky factor of $\M$ at that point, which is precisely the factor needed by the update of the next iteration. An accepted first trial thus costs nothing extra, and each rejected trial costs one additional factorization.

Another alternative is to use {\bf \emph{diminishing step-size} (dss)}, which however does not provide the best possible rate. As default dss-sequence, we use $\eta^k = \eta^{k-1}\,\big(1 - \epsilon\,\eta^{k-1}\big)$, where $\epsilon \in (0,1)$ is a given constant, and $\eta^0 <1/\epsilon$. We set $\eta^0=0.99$ and $\epsilon=0.05$ as default values. 

To conclude, a constant step size attains the best worst-case rate only if $\eta$ is tuned close to the unknown constant $\bar\eta$. The backtracking rule requires no tuning, retains the $\mathcal O(T^{-1/2})$ rate (its steps satisfy $\eta^k\ge\bar\eta/2$), and takes larger steps whenever the landscape permits.  Backtracking is therefore our recommended choice, while the diminishing schedule serves
mainly as a baseline for comparisons.

\section{SCA via damped fixed-point iterations}\label{sec:CL-FP}

The fixed-point characterization $\gam=[H(\gam)]_+$ of
Lemma~\ref{lem:fixed_point} suggests the iteration
$\gam^{k+1}=[H(\gam^k)]_+$. The first
observation is that this map is a \emph{diagonally scaled projected
gradient step}:
\begin{align}
    &[H(\gam)]_+
    =\arg\min_{\bm u\ge\bm0}\;
    \nabla\ell(\gam)^\top(\bm u-\gam)
    +\frac12\sum_{n=1}^N d_n^{2}\,(u_n-\gamma_n)^2 
    \label{eq:fp_scaled_gradient}
\end{align}
where $d_n = \nabla_n f(\gam)$ defined by \eqref{eq:partial_gammai_ncvx}. 
This result easily  follows because the quadratic problem in \eqref{eq:fp_scaled_gradient} is
separable in coordinates and its $n^{\text{th}}$ coordinate minimizer is
$[\gamma_n-\nabla_n\ell(\gam)/d_n^2]_+$, precisely the fixed point map $H_n(\gam)$. 
Thus we define damped fixed point iteration as
\beq \label{eq:CLFPiter}
\begin{aligned} 
\gam^{\iter+1} &=  \gam^{\iter} + \eta^k\, \mathbf{q}^k,  \ \mathbf{q}^k =[H(\gam^k)]_+ - \gam^k. 
\end{aligned} 
\eeq
The damped
iteration~\eqref{eq:CLFPiter} is itself a parallel SCA scheme: it minimizes, per coordinate, the strongly convex \emph{quadratic} surrogate in \eqref{eq:fp_scaled_gradient} (curvature $d_n^2$, gradient matching $\nabla\ell$ at $\gam^k$), whereas CL-SCA minimizes the nonquadratic surrogate $\ell_n$ of~\eqref{eq:tilde_f}.  This interpretation explains why, as shown below, the two algorithms obey the same convergence guarantee with the same constants.  The method \eqref{eq:CLFPiter}  was proposed in \cite[cf. Eq. (50)]{tan2025fast} under the name PSCA-ML-UD. Here, however, we derived the method as a damped fixed-point iteration and henceforth refer to the algorithm as CL-FP to emphasize this connection. In the next subsection we derive the stationary rate of CL-FP and discuss similarities and differences with the CL-SCA algorithm.  We emphasize that in \cite{tan2025fast}, the convergence rate and step size selection was not addressed and thus results of subsection~\ref{subsec:CL-FP-rate} are new. 

\subsection{Stationarity rate of CL-FP and comparison to CL-SCA} \label{subsec:CL-FP-rate}

We wish to point out that the  \emph{undamped} map $T(\gam)=[H(\gam)]_+$ is, in general, not a contraction, so that the classical fixed-point theory (Banach) does not apply and damping is not merely cosmetic.

\begin{theorem} \label{th:fp_rate}
Let $\mathbf{q}^k=[H(\gam^k)]_+-\gam^k$ and let $\underline d$,
$L_\nabla$, and $\bar\eta=\min\{1,\underline d^{\,2}/L_\nabla\}$ be
as in Lemma~\ref{lem:g_constants} and
Theorem~\ref{th:g_rate}. If $0<\eta^k\le\bar\eta$, then the
iterates of \eqref{eq:CLFPiter} satisfy $\gam^k\in\mathcal K_0$
and
\begin{equation}    \label{eq:fp_sufficient_decrease}
\ell(\gam^{k+1}) \le \ell(\gam^k) - \frac{\underline d^{\,2}}{2} \,\eta^k\|\bm q^k\|_2^2 , 
\end{equation}
and consequently the rate bound~\eqref{eq:g_rate}, the
best-iterate rate $\min_{k<T}\|\bm q^k\|_2=\mathcal O(T^{-1/2})$
for constant $\eta\in(0,\bar\eta]$, and when
$\inf_k\eta^k>0$, the convergence
$\|\bm q^k\|_2\to0$ with every accumulation point KKT-stationary, all
hold precisely.
\end{theorem}

\begin{proof} See Appendix~\ref{app:th:fp_rate}.
\end{proof}

The damped FP iterations, which we refer to as CL-FP algorithm and the earlier CL-SCA are two members of the same family: parallel SCA algorithms whose convex coordinate surrogates match $\nabla_n\ell$ at the current iterate and have curvature $2d_n^{2}$ (CL-SCA) and $d_n^{2}$ (CL-FP) at a stationary point; the CL-FP surrogate is quadratic, hence strongly convex with modulus $\underline d^{\,2}$ on $\mathcal K_0$. They also share identical guarantees: (a) the same safe damping bound $\bar\eta$, (b) the same sufficient-decrease constant $\underline d^{\,2}/2$, (c) the same $\mathcal O(T^{-1/2})$ rate, and (d)  the same per-iteration computational complexity. Their residuals, $\bm p = \hat\gam(\gam)-\gam$ in \eqref{eq:g_direction} and $\bm q =[H(\gam)]_+-\gam$ in \eqref{eq:CLFPiter}, respectively,  verify 
\begin{equation} \label{eq:qn_relateto_pn}
    q_n
    =(1+\sqrt{e_n/d_n})\,p_n
\end{equation}
whenever neither projection is active. Thus,  $|p_n|\le|q_n|$ always holds.  The residual norms,  $\Delta_{\mathrm{SCA}}(\gam)=\|\hat\gam(\gam)-\gam\|_2$ and $\Delta_{\mathrm{FP}}(\gam)=\|[H(\gam)]_+-\gam\|_2$, are equivalent stationarity measures:   $\Delta_{\mathrm{SCA}}(\gam) \le\Delta_{\mathrm{FP}}(\gam) \le(1+\sqrt{\bar r})\Delta_{\mathrm{SCA}}(\gam)$ with $\bar r=\max_{\gam\in\mathcal K_0}\max_n e_n/d_n<\infty$. The two maps differ,  however, in their \emph{local} behavior near a KKT point $\gam^\star$: there $e_n/d_n\to1$ on the support, so the FP step is asymptotically exactly \emph{twice} the SCA step, and the local Jacobians on the support (under strict complementarity) are 
\begin{equation*}
    \bm J^{\mathrm{FP}}_{\mathcal S}
    =\bm I-\bm D_{\mathcal S}
    [\nabla^2\ell]_{\mathcal S\mathcal S},
    \qquad
    \bm J^{\mathrm{SCA}}_{\mathcal S}
    =\bm I-\tfrac12\,\bm D_{\mathcal S}
    [\nabla^2\ell]_{\mathcal S\mathcal S}.
\end{equation*}
where $\bm D_{\mathcal S}=\diag\bigl(1/d_n(\gam^\star)^2\bigr)_{n\in\mathcal S}$ and $\mathcal S=\supp(\gam^\star)=\{\,n\in\{1,\ldots,N\}:\gamma_n^\star>0\,\}$. Therefore, the undamped SCA map is locally convergent whenever the eigenvalues of the scaled Hessian lie in $(0,4)$, while the undamped FP map requires the twice-smaller interval $(0,2)$: CL-SCA has double the local stability margin, where the full step $\eta^k=1$ is often  accepted by CL-SCA but destabilizes the FP iteration on the same data. This also confirms our empirical observation that CL-SCA tolerates roughly twice the damping factor of the FP iteration. 
 
Finally, since Theorem~\ref{th:fp_rate} holds with the same constants $\underline d$ and $L_\nabla$ as Theorem~\ref{th:g_rate}, the safe constant step size $\bar\eta$, the diminishing step-size rule, and the backtracking rule of Subsection~\ref{subsec:select_stepsize} apply to the FP iteration without modification. In practice, however, the FP displacement exceeds the SCA displacement by the factor $1+\sqrt{e_n/d_n}$ (as noted in \eqref{eq:qn_relateto_pn}), so a fixed step tuned for CL-SCA is typically too large for CL-FP, whereas the backtracking rule adapts to this automatically.

\subsection{Relation to the MM-based algorithm} \label{subsec:mm_vs_sca}

The  CLMM \cite{esfandiari2026mm} method majorizes \emph{both} terms of $\ell(\gam)$: the concave $\log|\M|$ by the same linearization used in \eqref{eq:tilde_f}, and the convex $g(\gam)$ by the matrix upper bound of \cite[Prop.~4]{sun2016robust}, applied after absorbing part of the noise floor into the signal powers so that the bound remains tight where $\gamma_n=0$. This requires a majorization constant $\nu \ge \lambda_{\max}(\A\A^\hop)$ (denoted $\eta$ in \cite{esfandiari2026mm} but  renamed here to avoid a clash with step size $\eta^k$).  After some algebra, and using notations in this paper,  the CLMM update rewrites as 
\begin{equation} \label{eq:CLMM_update}
\gamma_n^{k+1} = \Big[\, \gamma_n^{\iter} + \big(\gamma_n^{\iter} + \sigma^2/\nu \big)\Big( \sqrt{r_n^k} - 1 \Big) \Big]_+ ,
\end{equation}
which is precisely \eqref{eq:SCA_optim_gammai_rn} 
with the multiplier $1/d_n^k$ replaced by $\gamma_n^k + \sigma^2/\nu$. 
Together with \eqref{eq:MSBL_update} and \eqref{eq:gamma_i_sol2}, all four update maps considered in this paper are then members of a single family, each of the form $\gamma_n \mapsto \gamma_n + (\text{multiplier}) \times (\text{residual})$ before projection, as summarized in Table~\ref{tab:update_family}.

\begin{table}[!t]
\renewcommand{\arraystretch}{1.3}
\caption{The four update maps as one family: $\gamma_n \mapsto \gamma_n + (\text{multiplier}) \times (\text{residual})$, before projection onto $\gamma_n \ge 0$.
}\label{tab:update_family} \vspace{-4pt}
\centering
\begin{tabular}{lcc}
\hline
 & multiplier & residual  \\
\hline
EM of M-SBL \eqref{eq:MSBL_update} & $\gamma_n^2 d_n$ & $r_n - 1$  \\
CL-FP \eqref{eq:gamma_i_sol2}  & $(1+\sqrt{r_n})/d_n$ & $\sqrt{r_n} - 1$   \\
CL-SCA \eqref{eq:SCA_optim_gammai} & $1/d_n$ & $\sqrt{r_n} - 1$  \\
CLMM \eqref{eq:CLMM_update} & $\gamma_n + \sigma^2/\nu$ & $\sqrt{r_n} - 1$ \\
\hline
\end{tabular}
\end{table}

The residual vanishes exactly at $r_n = 1$, i.e. at $e_n = d_n$, so on the support all four maps stop at the same points. At the boundary $\gamma_n = 0$ they do not, and the multiplier is what decides it. For CL-FP, CL-SCA and CLMM the multiplier is bounded away from zero (($1+\sqrt{r_n})/d_n$, $1/d_n$, and $\gamma_n + \sigma^2/\nu \ge \sigma^2/\nu>0$, respectively), so a coordinate at $\gamma_n=0$ is left unchanged only when $r_n \le 1$, which is precisely the KKT condition~\eqref{eq:KKT}. Lemma~\ref{lem:g_first_order} therefore extends precisely to CLMM, whose fixed points are again the KKT points of \eqref{eq:orig_formulation}. The EM multiplier $\gamma_n^2 d_n$, in contrast, vanishes at $\gamma_n = 0$, so every coordinate with $\gamma_n=0$ is fixed irrespective of the sign of $\nabla_n \ell$: the fixed points of \eqref{eq:MSBL_update} strictly contain the KKT points, the spurious ones being those with $\gamma_n = 0$ and $e_n > d_n$. This is 
is also why \eqref{eq:MSBL_update} needs no projection: its multiplicative form cannot leave the nonnegative orthant. 

Moreover, the two multipliers of CL-SCA and CLMM can be compared exactly. By \eqref{eq:gamma_i_sol1_apu_2} we have $1/d_n = \gamma_n + 1/(\a_n^\hop \c_n)$, so both are of the form $\gamma_n$ plus a positive quantity, and since $\M_{\backslash n}\succeq \sigma^2 \I$ gives $\a_n^\hop \c_n \le \|\a_n\|^2/\sigma^2$ while $\nu \ge \lambda_{\max}(\A \A^\hop) \ge \|\a_n\|^2$,
\begin{equation} \label{eq:MM_damping}
\frac{\sigma^2}{\nu} \; \le \; \frac{\sigma^2}{\|\a_n\|^2} \; \le \; \frac{1}{\a_n^\hop \c_n} ,
\end{equation}
so that the CLMM multiplier never exceeds the CL-SCA one, at every coordinate and every iterate. In this precise sense, CLMM is an intrinsically \emph{damped} CL-SCA iteration, which explains why it is monotone with no step size while the parallel update \eqref{eq:SCA_smoothing2} requires the damping $\eta^k$: the damping of CLMM is built into its majorizer, and $\nu$ is its price. The gap in \eqref{eq:MM_damping} is exactly the looseness of $\nu$. With the choice $\nu = \tr(\A \A^\hop)$ used in \cite{esfandiari2026mm} and equal-norm atoms, $\nu = N\|\a_n\|^2$, starting from $\gam^0=\bm 0$, the CLMM step is thus shorter than the CL-SCA step by a factor $N$. This does not translate into a proportional loss of speed, however, since \eqref{eq:CLMM_update} is multiplicative in the shifted variable $\bar\gamma_n = \gamma_n + \sigma^2/\nu$, for which it reads $\bar\gamma_n^{k+1} = \sqrt{r_n^k}\, \bar\gamma_n^{k}$. A conservative $\nu$ therefore costs a logarithmic rather than a linear number of additional iterations. 

Two remarks complete the comparison. First, both surrogates contain the same linear term $d_n^k\gamma_n$ of \eqref{eq:tilde_f}, since both methods linearize the same concave term. They differ only in the convex term $g$, which CL-SCA keeps exact along the coordinate (via the Sherman--Morrison identity), whereas CLMM majorizes it to obtain a surrogate that is separable in all $N$ coordinates. The CL-SCA surrogate is thus tighter, but only along one
coordinate with the others frozen, whereas the CLMM surrogate is a joint majorizer. Second, the comparison above concerns the undamped map $\hat\gam(\gam^k)$, i.e., \eqref{eq:SCA_optim_gammai} with $\eta^k=1$, not
the damped iteration of Algorithm~\ref{alg:SCA}.

\subsection{Convergence analysis and a new Hybrid method}
\label{subsec:convergence}

Convergence is assessed with two complementary measures. The \emph{objective gap} $\ell(\gam^k)-\ell^\star$, where $\ell^\star$ is the smallest objective value recorded over a $1500$-iteration reference run, is plotted on a logarithmic axis to resolve the behavior near stationarity: a locally linearly convergent method appears as a straight line whose slope is its contraction rate. Since the raw objective is used, the gap also exposes the non-monotone transient of the diminishing step-size variants. The \emph{normalized running-minimum loss} $\ell_{\mathrm{norm}}^k=\min_{s\le k}\ell(\gam^s)/\ell(\gam^0)$ starts from one for every run and is monotone for all methods. It measures time-to-target and is the view in which iteration counts and wall-clock times are reported.

All methods are initialized at $\gam^0=\mathbf{1}$, since zero coordinates remain fixed under the EM update, and a common initialization ensures that methods converging to the same stationary point converge to the same level in both figures. We compare CWO~\cite[Alg.~1]{fengler2021non}, EM-based solver of M-SBL~\cite{wipf2007empirical}, and CLMM  \cite{esfandiari2026mm} which uses $\nu = \| \A \|_{\mathrm{F}}^2$. We also consider a variant CLMM-2 which uses $\nu= \| \A \|_2^2$ (instead of the Frobenius norm) which by  \eqref{eq:MM_damping} and $\nu \geq \lambda_{\max}(\A \A^\hop) = \| \A \|_2^2 $ provides the largest step size and thus a better local rate. Finally, we consider  the proposed variants: CL-SCA-dss/CL-FP-dss use the diminishing step-size sequence with $\eta^0=0.99$ and $\epsilon=0.05$, CL-SCA-fss uses $\eta=0.6$  and CL-FP-fss uses $\eta=0.35$.  These were hand-tuned step-sizes (chosen as large as possible) that worked for our set-up with dimension $M=20$.  The hybrid CL-Hyb applies CL-SCA-fss for the first five iterations and CL-FP-fss thereafter, with the same tuned step sizes. CL-SCA-bt and CL-Hyb-bt are the corresponding tuning-free variants: every step size is selected by the backtracking rule of Subsection~\ref{subsec:select_stepsize}, and CL-Hyb-bt switches from the SCA to the FP proposal after five iterations, as CL-Hyb does. Since near a KKT point the FP displacement is about twice the SCA displacement, an FP step $\eta$ acts locally like an SCA step $2\eta$; at the switch, CL-Hyb-bt therefore halves the last accepted step size, so that its first FP trial has the same effective length as the last SCA step. In order to be fair, all methods are implemented by ourselves in Python/NumPy with identical linear-algebra routines and optimized for best performance. For CWO we also use the random permutation strategy which is reported to improve its speed \cite{chen2019covariance}.

{\it Simulation setting:} For each trials, $\A$ is a random Bernoulli matrix, with elements generated as  $[\A]_{mn} \in \{ \pm \tfrac{1}{\sqrt{2}} \pm \jmath\tfrac{1}{\sqrt{2}} \}$, so that $\|\a_n\|_2^2=M$ while the support $\mathcal M$ is drawn
uniformly without replacement from $\{1,\ldots,N\}$.  The noise variance  is fixed $\sigma^2=1$ and the $K$-nonzero signal powers of the Gaussian sources are uniformly distributed between \([-15,0]\) in dB scale and generated indpendently for each trial. 

Fig.~\ref{fig:loss_curve} reports both measures for $K=30$, $M=20$, $N=300$, and $Q=50$ over $20$ MC runs. The per-run target $\ell_{\mathrm{ref}}$ in the bottom panel is the objective value reached by ten CWO cycles. This time-to-target metric compares progress in objective value only: ten CWO cycles are a heuristic budget that does not imply stationarity.

Several observations can be drawn. First, the two baselines are far slower than the proposed methods. The EM iterations of M-SBL reduce the gap by less than a factor of $20$ over $120$ iterations, from $60$ to $3.2$, and never reach $\ell_{\mathrm{ref}}$ within $250$ iterations; CLMM decays faster but still slowly (gap $6.3\times10^{-3}$ at $k=120$) and does not reach the target either. CLMM-2 improves over CLMM and achieves gap $2.53 \times10^{-3}$ at $k=120$.  By contrast, the fixed-step and backtracking variants drive the gap below $10^{-5}$ within $55$ iterations, and the diminishing-step variants within $100$.

Second, the fixed step size gives a substantial gain over the diminishing schedule. At $k=50$ the median gap of CL-SCA-fss is $1.8\times10^{-5}$ against $2.6\times10^{-4}$ for CL-SCA-dss, and that of CL-FP-fss is $8.4\times10^{-6}$ against $8.6\times10^{-5}$ for CL-FP-dss, reductions by factors of $14$ and $10$ that exceed two orders of magnitude by $k=100$. The dss curves are also markedly non-monotone during the first ten iterations, which the running-minimum view hides by construction. Both effects agree with Section~\ref{sec:convergence}: a diminishing sequence shrinks the accumulated step length that drives the convergence rate.

Third, the two fixed-step methods are complementary. At $k=10$ the CL-SCA-fss gap is an order of magnitude smaller than that of CL-FP-fss ($3.3\times10^{-1}$ versus $3.8$), whereas from $k\approx 35$ the ordering reverses. This agrees with the relation between their displacements: when the projection is inactive, the FP step equals $(1+\sqrt{e_n/d_n})$ times the SCA step. Far from stationarity $e_n/d_n$ can be large, so the FP proposal is more aggressive and needs a smaller damping factor ($\eta=0.35$ versus $0.6$); near stationarity $e_n/d_n\to 1$ on the support, and the FP step becomes about
twice the SCA step.

This motivates CL-Hyb, which uses the more stable SCA update for five iterations and then switches to FP. Since both updates satisfy the same sufficient-decrease inequality with the same constants
({\it cf.}~Section~\ref{sec:convergence}), the convergence guarantees carry over. CL-Hyb tracks CL-SCA-fss through the transient and then inherits the faster FP decay: at $k=50$ its gap is $3.9\times10^{-6}$, against $8.4\times10^{-6}$ for CL-FP-fss and $1.8\times10^{-5}$ for CL-SCA-fss.

Fourth, backtracking yields a further large gain without any tuning. After a slower start, where the sufficient-decrease test enforces small steps, CL-Hyb-bt attains the smallest gap of all methods from $k\approx 15$. At $k=50$ the gaps of CL-SCA-bt and CL-Hyb-bt are $5.8\times10^{-7}$ and $2.1\times10^{-8}$, i.e., $31$ and $185$ times smaller than those of CL-SCA-fss and CL-Hyb, and they reach the numerical floor of the gap ($\approx 5\times10^{-12}$) by $k\approx 100$ and $k\approx 80$, respectively, against $k\approx 120$ for CL-Hyb. The accepted steps are typically larger than the hand-tuned constants. 
 
Fifth, the running-minimum panel translates these gains into runtime. CL-Hyb-bt reaches $\ell_{\mathrm{ref}}$ in $30$ iterations and $3.0$~ms on average, and CL-FP-fss in $44$ iterations and $3.9$~ms, whereas M-SBL and CLMM do not reach it within $250$ iterations ($>28.5$~ms and $>21.8$~ms). CWO needs only ten cycles but $8.4$~ms, so CL-Hyb-bt in this setting attains the same objective level in near one third of the time despite using three times as many iterations. The reason is the cost per iteration: CL-SCA/CL-FP update the full vector $\hat\gam$ in one vectorized parallel step, whereas CWO cycles through the $N$ coordinates sequentially.

Finally, we computed the probability of missdetection ($P_{\text{MD}}$) under the top-$K$ rule, defined as the MC average of $|\mathcal{M} \setminus \hat{\mathcal{M}}|/K$ All CL-SCA, CL-FP and CL-Hyb variants and CWO give $P_{\mathrm{MD}}=0.267$, consistent with their converging to the same stationary point, whereas M-SBL stalls at a higher loss and yields $P_{\mathrm{MD}}=0.293$. CLMM also attains $0.267$ although it does not reach $\ell_{\mathrm{ref}}$, indicating that the support is identified well before full convergence.

\begin{figure}[!t]
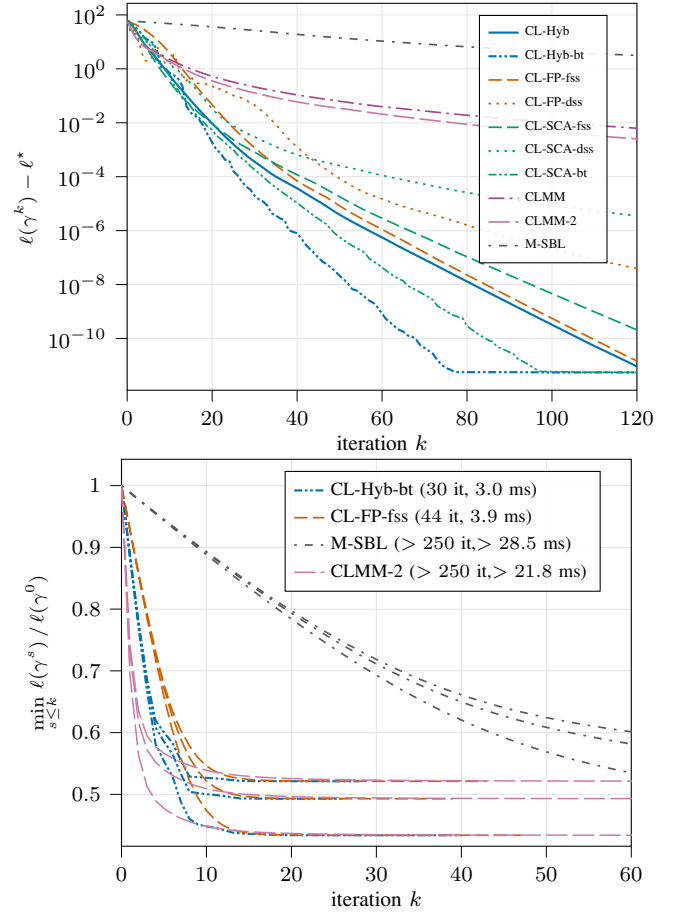

\centerline{\input{tikz/gap_c_K=30_M=50_N=300_L=20_MC=20_20261013_2354.tex}}
\centerline{\input{tikz/runmin_K=30_M=50_N=300_L=20_MC=20_20261013_2354}}
\vspace{-3pt}
\caption{Convergence for $M=20$, $K=30$, $Q=50$, $N=300$ ($20$ MC runs)  Top: median objective gap $\ell(\gam^k)-\ell^\star$.
Bottom: normalized running-minimum loss for representative runs, drawn until reaching the target $\ell_{\mathrm{ref}}$ (the loss after ten CWO cycles) or
$250$ iterations. The legend gives the average iterations and time-to-target.   ``$>$'' marks methods that did not reach it.  CWO required $8.4$~ms for its ten cycles. 
For better visualisation we truncate range of $k$ ($x$-axis) to $k\leq 60$ at the bottom panel. } 
\label{fig:loss_curve}
\end{figure}

\section{Application to JADCE in GFRA} \label{sec:source} 

JADCE in GFRA  is a central problem in  mMTC, where many devices with sporadic traffic access the uplink. Recent work shows that CL methods can outperform conventional CS approaches \cite{chen2019covariance,fengler2021non,chen2021phase,marata2025activity,wang2025robust}. In contrast to CS methods, whose performance is limited when the number of active devices exceeds the pilot length, CL methods can relax this limitation and achieve comparable accuracy with substantially shorter pilots \cite{fengler2021non,chen2021phase}.

We consider an uplink, single-cell massive random access scenario with $N$ single-antenna machine-type devices (MTDs) communicating with a BS equipped with $Q$ antennas.  For device identification, each device $n$ is preassigned a unique pilot (signature) sequence $\mathbf{a}_n = (a_{1n}, \ldots, a_{Mn})^\top$ which is known at the BS. The device traffic is assumed to be sporadic, that is, in each coherence interval only $K \ll N$ devices are randomly active. In mMTC settings, the total number of devices is typically much larger than the pilot length ($N \gg M$).

The received data follow the MMV model~\eqref{eq:MMVmodel}, where  $\Y \in \C^{M\times Q}$ contains the observations over $M$ pilot symbols and $Q$ BS antennas, and $\A=(\a_1 \cdots  \, \a_N)\in\C^{M\times N}$ collects the device pilot sequences. The effective channel matrix $\X=(\tilde{\x}_1 \, \ldots \, \tilde{\x}_N)^\top\in\C^{N\times Q}$ has rows
\beq
\tilde{\x}_n=\sqrt{\gamma_n}\mathbf{h}_n,\qquad
\gamma_n=\alpha_n\beta_n\rho_n,
\label{eq111}
\eeq
where $\mathbf{h}_n\in\C^Q$ is the small-scale channel vector, $\alpha_n\in\{0,1\}$ is the activity indicator, $\beta_n>0$ is the large-scale fading coefficient (LSFC), and $\rho_n$ is the transmit power. Since
only $K\ll N$ devices are active, $\X$ is $K$-rowsparse, with support $\mathcal{M}=\{n\in[\![N]\!]:\alpha_n=1\}.$ Under uncorrelated Rayleigh fading, the entries of $\mathbf{h}_n$ are modeled as i.i.d. $\mathcal{CN}(0,1)$. Thus, each received antenna snapshot has covariance as in \eqref{eq:M} and the corresponding scaled negative LLF  is the CL objective in~\eqref{eq:conv+nonconv}. JADCE is then performed using Algorithm~\ref{alg:SBL-JADCE} with any CL solver in Step~1. 

The number of active devices $|\mathcal{M}|=K$ is assumed to be fixed  (non-random, deterministic variable) and known. Hence  the estimated support $\hat{\mathcal M}$ is obtained by the top-$K$  rule, yielding $|\hat{\mathcal M}|=K$. Thus it is sufficient to use  $P_{\text{MD}}$ as the performance metric  since $|\mathcal M\setminus\hat{\mathcal M}| =|\hat{\mathcal M}\setminus\mathcal M|$,  and thus (empirical) probability of false alarm ($P_{\text{FA}})$ contain the same information up to normalization. 

Simulation setting follows that of \autoref{subsec:convergence}, i.e., random Bernoulli matrix $\A$, $\gamma_n \sim \mathrm{Unif}[-15,0]$ in dB scale, $\sigma^2=1$, and $\mathcal M$ (with $|\mathcal M|=K$) is drawn  uniformly without replacement from $\{1,\ldots,N\}$ for each MC trial.  Note that we assume no power-control, i.e., $\rho_n=1$ in \eqref{eq111}.  To evaluate the JADCE framework in Algorithm~\ref{alg:SBL-JADCE}, we compare different solvers for Step~1: the proposed CL-Hyb-bt,  the EM solver of M-SBL~\cite{wipf2007empirical}, CLMM-2, CWO~\cite[Alg.~1]{fengler2021non}, and
CL-MP~\cite[Alg.~1]{marata2025activity}. For all methods, channel estimates are computed by Step~4 of Algorithm~\ref{alg:SBL-JADCE}.  For CL-MP, knowledge of $K$ is used to stop after $K$ selected atoms. We terminate each algorithm using \eqref{eq:terminate_CL} with $\delta=10^{-2}$.   This criterion is algorithm-independent: it is a function of the current iterate $\gam^k$ alone, and it vanishes if and only if $\gam^k$ satisfies the KKT conditions. Consequently, all methods are stopped once they attain the same level of proximity to the KKT conditions, which makes the comparison of  running times fair. The maximum number of iterations is capped to $500$ for all methods.  The initial start is $\gam^0 = \mathbf{1}$ for M-SBL and $\gam^{0}= \mathbf{0}$ for other methods.

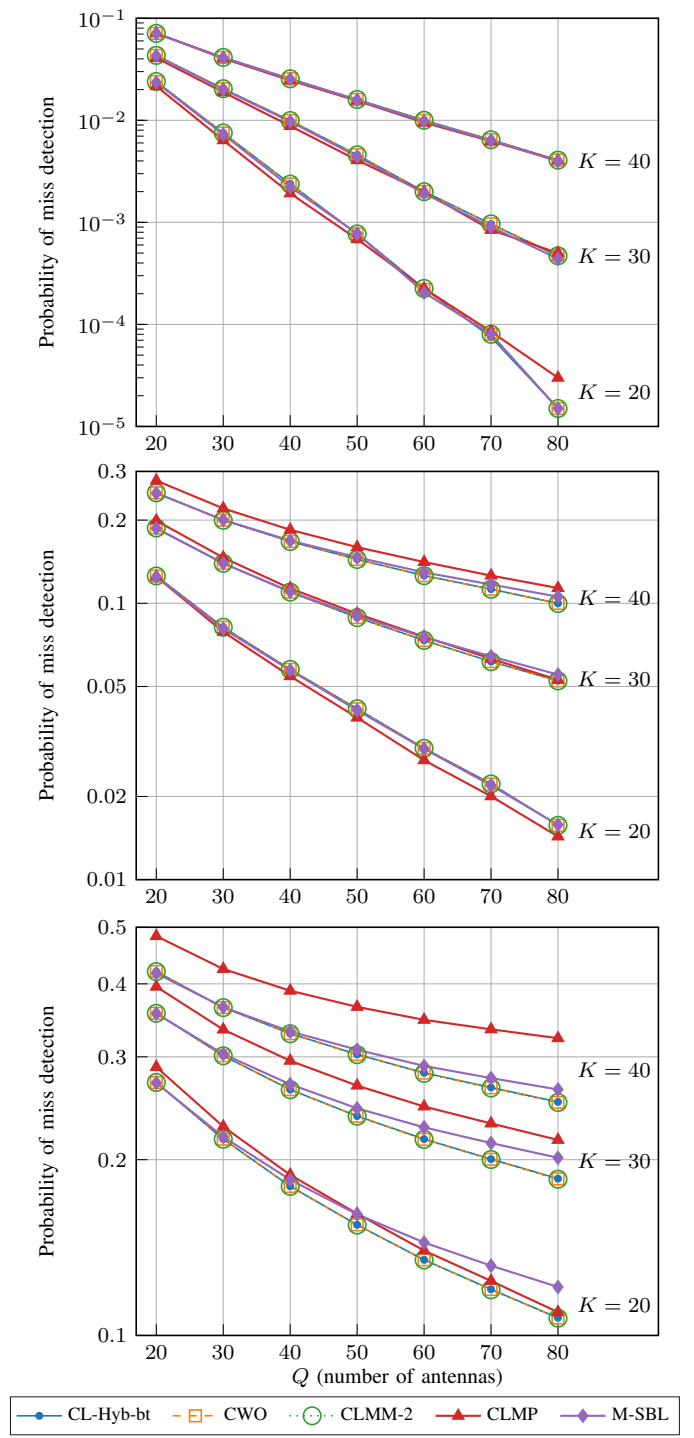
\begin{figure}[!t]   
\centering
\begin{tikzpicture}

\tikzset{
  hyb/.style ={semithick, steelblue31119180,     solid,   mark=*,        mark size=1.2},
  cwo/.style ={semithick, darkorange25512714,    dashed,  mark=square,   mark size=2.2,
               mark options={solid, fill=none}},
  clmm/.style={semithick, forestgreen4416044,    dotted,  mark=o,        mark size=3.3,
               mark options={solid, fill=none}},
  clmp/.style={thick,     crimson2143940,        solid,   mark=triangle*,mark size=2.2},
  msbl/.style={thick,     mediumpurple148103189, solid,   mark=diamond*, mark size=2.2},
}

\begin{axis}[
  width=0.78\columnwidth,
  height=5.4cm,
     scale only axis,
  axis line style={semithick},
     yticklabel style={text width=2.2em, align=right},
  xmin=17, xmax=95,                 
  xtick={20,30,40,50,60,70,80},
  xmajorgrids,
  ymode=log,
  ymin=1e-5, ymax=0.1,
  ytick={1e-5,1e-4,1e-3,1e-2,1e-1},
  ylabel={Probability of miss detection},
  ymajorgrids,
  grid style={darkgray176},
  tick pos=left,
  tick style={color=black},
  tick label style={font=\footnotesize},
  xlabel style={yshift=3pt, font=\footnotesize},
  ylabel style={yshift=-3pt, font=\footnotesize},
]

\addplot[hyb] coordinates {(20,0.02396) (30,0.007505) (40,0.002385) (50,0.00077) (60,0.000225) (70,7.5e-5) (80,1.5e-5)};
\addplot[hyb] coordinates {(20,0.043273) (30,0.020417) (40,0.00995) (50,0.00458) (60,0.0020067) (70,0.00096333) (80,0.00046667)};
\addplot[hyb] coordinates {(20,0.071123) (30,0.041417) (40,0.025555) (50,0.015998) (60,0.010018) (70,0.0064475) (80,0.0040625)};

\addplot[cwo] coordinates {(20,0.02396) (30,0.007505) (40,0.002375) (50,0.00077) (60,0.00022) (70,8e-5) (80,1.5e-5)};
\addplot[cwo] coordinates {(20,0.043287) (30,0.02041) (40,0.00995) (50,0.0046) (60,0.0020067) (70,0.00096667) (80,0.00047667)};
\addplot[cwo] coordinates {(20,0.071113) (30,0.041385) (40,0.025545) (50,0.01602) (60,0.010008) (70,0.00645) (80,0.00408)};

\addplot[clmm] coordinates {(20,0.02398) (30,0.00752) (40,0.00237) (50,0.00077) (60,0.000225) (70,8e-5) (80,1.5e-5)};
\addplot[clmm] coordinates {(20,0.043283) (30,0.020397) (40,0.0099467) (50,0.0045833) (60,0.0019967) (70,0.00096667) (80,0.00047)};
\addplot[clmm] coordinates {(20,0.071093) (30,0.041382) (40,0.02553) (50,0.016) (60,0.010008) (70,0.006445) (80,0.00407)};

\addplot[clmp] coordinates {(20,0.02162) (30,0.00638) (40,0.001925) (50,0.000685) (60,0.000225) (70,8.5e-5) (80,3e-5)};
\addplot[clmp] coordinates {(20,0.04089) (30,0.01887) (40,0.00881) (50,0.0040667) (60,0.0019867) (70,0.00084333) (80,0.0005)};
\addplot[clmp] coordinates {(20,0.071443) (30,0.04038) (40,0.024628) (50,0.015543) (60,0.009505) (70,0.0062425) (80,0.004035)};

\addplot[msbl] coordinates {(20,0.023455) (30,0.007275) (40,0.002225) (50,0.000775) (60,0.000205) (70,8e-5) (80,1.5e-5)};
\addplot[msbl] coordinates {(20,0.04275) (30,0.020053) (40,0.0097233) (50,0.0044167) (60,0.00192) (70,0.00090333) (80,0.00044)};
\addplot[msbl] coordinates {(20,0.07071) (30,0.04109) (40,0.025318) (50,0.01579) (60,0.00978) (70,0.00635) (80,0.0039675)};

\node[anchor=west, font=\footnotesize] at (axis cs:81.5,2.2e-5)  {$K=20$};
\node[anchor=west, font=\footnotesize] at (axis cs:81.5,4.7e-4)  {$K=30$};
\node[anchor=west, font=\footnotesize] at (axis cs:81.5,4.05e-3) {$K=40$};

\end{axis}
\end{tikzpicture}\\[-1pt]
\begin{tikzpicture}

\begin{axis}[
  width=0.78\columnwidth,
  height=5.4cm,
  scale only axis,
  axis line style={semithick},
  xmin=17, xmax=95,                 
  xtick={20,30,40,50,60,70,80},
  xmajorgrids,
  ymode=log,
  ymin=0.01, ymax=0.3,
  ytick={0.01,0.02,0.05,0.1,0.2,0.3},
  log ticks with fixed point,
     yticklabel style={/pgf/number format/.cd, fixed, precision=2,
                     /tikz/.cd, text width=2.2em, align=right},
  ylabel={Probability of miss detection},
  ymajorgrids,
  grid style={darkgray176},
  tick pos=left,
  tick style={color=black},
  tick label style={font=\footnotesize},
  xlabel style={yshift=3pt, font=\footnotesize},
  ylabel style={yshift=-3pt, font=\footnotesize},
]

\addplot[hyb] coordinates {(20,0.12554) (30,0.08197) (40,0.05767) (50,0.04154) (60,0.02989) (70,0.02224) (80,0.01572)};
\addplot[hyb, forget plot] coordinates {(20,0.18730) (30,0.13958) (40,0.10984) (50,0.08880) (60,0.07342) (70,0.06156) (80,0.05243)};
\addplot[hyb, forget plot] coordinates {(20,0.25104) (30,0.19989) (40,0.16737) (50,0.14423) (60,0.12595) (70,0.11254) (80,0.09996)};

\addplot[cwo] coordinates {(20,0.12555) (30,0.08193) (40,0.05764) (50,0.04151) (60,0.02986) (70,0.02217) (80,0.01572)};
\addplot[cwo, forget plot] coordinates {(20,0.18731) (30,0.13957) (40,0.10981) (50,0.08879) (60,0.07338) (70,0.06154) (80,0.05240)};
\addplot[cwo, forget plot] coordinates {(20,0.25101) (30,0.19988) (40,0.16739) (50,0.14412) (60,0.12588) (70,0.11246) (80,0.09989)};

\addplot[clmm] coordinates {(20,0.12555) (30,0.08192) (40,0.05764) (50,0.04154) (60,0.02987) (70,0.02219) (80,0.01572)};
\addplot[clmm, forget plot] coordinates {(20,0.18727) (30,0.13953) (40,0.10977) (50,0.08873) (60,0.07340) (70,0.06148) (80,0.05238)};
\addplot[clmm, forget plot] coordinates {(20,0.25101) (30,0.19978) (40,0.16735) (50,0.14409) (60,0.12589) (70,0.11243) (80,0.09984)};

\addplot[clmp] coordinates {(20,0.12483) (30,0.07883) (40,0.05448) (50,0.03858) (60,0.02705) (70,0.02003) (80,0.01433)};
\addplot[clmp, forget plot] coordinates {(20,0.19914) (30,0.14633) (40,0.11318) (50,0.09168) (60,0.07569) (70,0.06301) (80,0.05285)};
\addplot[clmp, forget plot] coordinates {(20,0.27804) (30,0.22044) (40,0.18452) (50,0.15961) (60,0.14104) (70,0.12597) (80,0.11353)};

\addplot[msbl] coordinates {(20,0.12446) (30,0.08092) (40,0.05723) (50,0.04101) (60,0.02972) (70,0.02194) (80,0.01579)};
\addplot[msbl, forget plot] coordinates {(20,0.18662) (30,0.13953) (40,0.11025) (50,0.09009) (60,0.07537) (70,0.06428) (80,0.05519)};
\addplot[msbl, forget plot] coordinates {(20,0.25027) (30,0.20039) (40,0.16862) (50,0.14667) (60,0.12937) (70,0.11699) (80,0.10562)};

\node[anchor=west, font=\footnotesize] at (axis cs:81.5,0.0150) {$K=20$};
\node[anchor=west, font=\footnotesize] at (axis cs:81.5,0.0535) {$K=30$};
\node[anchor=west, font=\footnotesize] at (axis cs:81.5,0.1050) {$K=40$};

\end{axis}
\end{tikzpicture}\\[0pt]
\begin{tikzpicture}

\begin{axis}[
  width=0.78\columnwidth,
  height=5.4cm,
  scale only axis,
  axis line style={semithick},
 xlabel={$Q$ (number of antennas)},
  xmin=17, xmax=95,                 
  xtick={20,30,40,50,60,70,80},
  xmajorgrids,
  ymode=log,
  ymin=0.1, ymax=0.5,
  ytick={0.1,0.2,0.3,0.4,0.5},
  log ticks with fixed point,
     yticklabel style={/pgf/number format/.cd, fixed, precision=1,
                     /tikz/.cd, text width=2.2em, align=right},
  ylabel={Probability of miss detection},
  ymajorgrids,
  grid style={darkgray176},
  tick pos=left,
  tick style={color=black},
  tick label style={font=\footnotesize},
  xlabel style={yshift=3pt, font=\footnotesize},
  ylabel style={yshift=-3pt, font=\footnotesize},
   legend to name=Pmdlegend,
   legend style={font=\scriptsize, draw=white!15!black, inner sep=2pt,
                 /tikz/every even column/.append style={column sep=6pt}},
  legend columns=5,
]

\addplot[hyb] coordinates {(20,0.27115) (30,0.21684) (40,0.18011) (50,0.15468) (60,0.13476) (70,0.11997) (80,0.10722)};
\addlegendentry{CL-Hyb-bt}
\addplot[hyb, forget plot] coordinates {(20,0.35621) (30,0.30134) (40,0.26370) (50,0.23752) (60,0.21692) (70,0.20048) (80,0.18560)};
\addplot[hyb, forget plot] coordinates {(20,0.41990) (30,0.36433) (40,0.32888) (50,0.30295) (60,0.28168) (70,0.26587) (80,0.25118)};

\addplot[cwo] coordinates {(20,0.27107) (30,0.21683) (40,0.18001) (50,0.15469) (60,0.13475) (70,0.11990) (80,0.10713)};
\addlegendentry{CWO}
\addplot[cwo, forget plot] coordinates {(20,0.35622) (30,0.30132) (40,0.26364) (50,0.23738) (60,0.21663) (70,0.20026) (80,0.18544)};
\addplot[cwo, forget plot] coordinates {(20,0.42001) (30,0.36434) (40,0.32871) (50,0.30284) (60,0.28141) (70,0.26553) (80,0.25074)};

\addplot[clmm] coordinates {(20,0.27112) (30,0.21682) (40,0.18003) (50,0.15458) (60,0.13470) (70,0.11973) (80,0.10704)};
\addlegendentry{CLMM-2}
\addplot[clmm, forget plot] coordinates {(20,0.35617) (30,0.30123) (40,0.26353) (50,0.23740) (60,0.21666) (70,0.20012) (80,0.18520)};
\addplot[clmm, forget plot] coordinates {(20,0.41979) (30,0.36425) (40,0.32863) (50,0.30257) (60,0.28109) (70,0.26530) (80,0.25063)};

\addplot[clmp] coordinates {(20,0.28817) (30,0.22787) (40,0.18838) (50,0.16162) (60,0.13977) (70,0.12400) (80,0.10962)};
\addlegendentry{CLMP}
\addplot[clmp, forget plot] coordinates {(20,0.39554) (30,0.33410) (40,0.29538) (50,0.26781) (60,0.24664) (70,0.23068) (80,0.21619)};
\addplot[clmp, forget plot] coordinates {(20,0.48288) (30,0.42418) (40,0.38937) (50,0.36536) (60,0.34712) (70,0.33438) (80,0.32289)};

\addplot[msbl] coordinates {(20,0.27074) (30,0.21879) (40,0.18478) (50,0.16124) (60,0.14428) (70,0.13166) (80,0.12107)};
\addlegendentry{M-SBL}
\addplot[msbl, forget plot] coordinates {(20,0.35590) (30,0.30335) (40,0.26920) (50,0.24485) (60,0.22717) (70,0.21353) (80,0.20150)};
\addplot[msbl, forget plot] coordinates {(20,0.41702) (30,0.36485) (40,0.33153) (50,0.30853) (60,0.28938) (70,0.27589) (80,0.26375)};

\node[anchor=west, font=\footnotesize] at (axis cs:81.5,0.113) {$K=20$};
\node[anchor=west, font=\footnotesize] at (axis cs:81.5,0.199) {$K=30$};
\node[anchor=west, font=\footnotesize] at (axis cs:81.5,0.285) {$K=40$};

\end{axis}
\end{tikzpicture}\\[0pt]
\ref{Pmdlegend}
\caption{Probability of miss detection vs.\ the number of antennas $Q$ when $M=20$ (bottom panel), $M=30$ (middle panel) and $M=50$ (top panel). Number of MC trials: $10000$; stopping tolerance for all algorithms: $\delta = 10^{-2}$.}
\label{fig:AD1}
\end{figure}

First we study the impact of the pilot length $M$, number of BS antennas $Q$, and sparsity level $ K$ on JADCE performance. Fig.~\ref{fig:AD1} plots $P_{\mathrm{MD}}$ versus $Q$ for $K\in\{20,30,40\}$ and  $M=20,30,50$ from bottom to top.  We observe that CL-MP performs well for small $K$, but degrades as $K$ increases, especially for $M=20$ and $M=30$. Overall, CL-Hyb-bt, CLMM-2, and CWO achieve nearly identical missed-detection performance, while EM-based M-SBL is slightly worse, with a more visible gap for short pilots, e.g., $M=20$.

Fig.~\ref{fig:CPU} reports the  runtimes for $M=20$ and $M=50$ averaged over $Q$ and MC trials. CL-Hyb-bt is consistently faster than CWO, roughly by a factor of two.  CL-MP is the fastest: it is up to about three and a half times faster for $M = 20$ and $K=20$, and about twice as fast for $M = 20$ and $K=40$.  This is expected, since CL-MP is a greedy method that terminates after exactly $K$ iterations when $K$ is known. Finally, EM-based M-SBL is more than one order of magnitude slower than the proposed CL-Hyb-bt method, mainly because it requires several hundred iterations to converge.

\begin{figure}[!t]
\vspace{-2pt}
\centerline{\hspace{-20pt}\definecolor{crimson2143940}{RGB}{214,39,40}
\definecolor{darkgray176}{RGB}{176,176,176}
\definecolor{darkorange25512714}{RGB}{255,127,14}
\definecolor{forestgreen4416044}{RGB}{44,160,44}
\definecolor{mediumpurple148103189}{RGB}{148,103,189}
\definecolor{steelblue31119180}{RGB}{31,119,180}

\begin{tikzpicture}

\begin{axis}[
    width=0.39\columnwidth,
  height=3.9cm,
  scale only axis,
  ybar=0pt,                    
  bar width=1.5,               
  xtick={20,30,40},
  enlarge x limits=0.22,       
  ymin=0, ymax=0.095,          
  ytick={0,0.02,0.04,0.06,0.08},
  ymajorgrids,
  y grid style={darkgray176},
  tick pos=left,
  xtick style={color=black},
  ytick style={color=black},
  y tick label style={font=\notsotiny},
  x tick label style={font=\footnotesize},
  yticklabel style={/pgf/number format/.cd, fixed, fixed zerofill, precision=2},
  scaled y ticks=false,
  xlabel={$K$ (number of active devices)},
  ylabel={Running time [s]},
  xlabel style={yshift=3pt, font=\footnotesize},
  ylabel style={yshift=-3pt, font=\footnotesize},
  title={(a) $M=20$},
  title style={yshift=-7pt, font=\footnotesize},
  legend to name=CPUlegend,
  legend columns=5,
  legend style={font=\footnotesize, draw=white!15!black, inner sep=2pt,
                /tikz/every even column/.append style={column sep=8pt}},
  legend image code/.code={%
    \draw[#1,draw=none] (0cm,-0.5ex) rectangle (1.2ex,0.5ex);
  },
]

\addplot[draw=none, fill=steelblue31119180]
  coordinates {(20,0.00554122) (30,0.00579648) (40,0.00602830)};
\addlegendentry{CL-Hyb-bt}

\addplot[draw=none, fill=darkorange25512714]
  coordinates {(20,0.01312978) (30,0.01429363) (40,0.01537268)};
\addlegendentry{CWO}

\addplot[draw=none, fill=forestgreen4416044]
  coordinates {(20,0.04204229) (30,0.04552858) (40,0.04661181)};
\addlegendentry{CLMM-2}

\addplot[draw=none, fill=crimson2143940]
  coordinates {(20,0.00160529) (30,0.00238905) (40,0.00325476)};
\addlegendentry{CLMP}

\addplot[draw=none, fill=mediumpurple148103189]
  coordinates {(20,0.08557656) (30,0.08666943) (40,0.09010133)};
\addlegendentry{MSBL}

\end{axis}
\end{tikzpicture}\begin{tikzpicture}

\definecolor{crimson2143940}{RGB}{214,39,40}
\definecolor{darkgray176}{RGB}{176,176,176}
\definecolor{darkorange25512714}{RGB}{255,127,14}
\definecolor{forestgreen4416044}{RGB}{44,160,44}
\definecolor{mediumpurple148103189}{RGB}{148,103,189}
\definecolor{steelblue31119180}{RGB}{31,119,180}

\begin{axis}[
  width=0.4\columnwidth,
  height=3.9cm,
  scale only axis,
  ybar=0pt,                    
  bar width=1.5,               
  xtick={20,30,40},
  enlarge x limits=0.22,       
    ymin=0, ymax=0.37,            
  ytick={0,0.1,0.2,0.3,0.4},
  ymajorgrids,
  y grid style={darkgray176},
  tick pos=left,
  xtick style={color=black},
  ytick style={color=black},
  y tick label style={font=\notsotiny},
  x tick label style={font=\footnotesize},
  yticklabel style={/pgf/number format/.cd, fixed, fixed zerofill, precision=1},
  scaled y ticks=false,
  xlabel={$K$ (number of active devices)},
  xlabel style={yshift=3pt, font=\footnotesize},
  ylabel style={yshift=-3pt, font=\footnotesize},
  title={(b) $M=50$},            
  title style={yshift=-7pt, font=\footnotesize},
  legend image code/.code={%
    \draw[#1,draw=none] (0cm,-0.5ex) rectangle (1.2ex,0.5ex);
 },
]

\addplot[draw=none, fill=steelblue31119180]
  coordinates {(20,0.01119373) (30,0.01260869) (40,0.01341559)};

\addplot[draw=none, fill=darkorange25512714]
  coordinates {(20,0.02379183) (30,0.02706003) (40,0.02849120)};

\addplot[draw=none, fill=forestgreen4416044]
  coordinates {(20,0.09655262) (30,0.11783863) (40,0.13351333)};

\addplot[draw=none, fill=crimson2143940]
  coordinates {(20,0.00308735) (30,0.00488195) (40,0.00641794)};

\addplot[draw=none, fill=mediumpurple148103189]
  coordinates {(20,0.34158264) (30,0.35904873) (40,0.35892386)};

\end{axis}
\end{tikzpicture}}
\vspace{-2pt}
\centerline{\ref{CPUlegend}}
\vspace{-2pt}
\caption{Running time versus the number of active devices $K$ for $M=20$ (a) and $M=50$ (b), averaged over $Q\in\{20,30,\ldots,80\}$ and 1000 MC trials per $Q$. All optimization-based methods use the stopping criterion~\eqref{eq:terminate_CL} with $\delta=10^{-2}$.}\label{fig:CPU}

\medskip

\centerline{
\begin{tikzpicture}

\definecolor{darkgray176}{RGB}{176,176,176}
\definecolor{forestgreen4416044}{RGB}{44,160,44}
\definecolor{steelblue31119180}{RGB}{31,119,180}

\begin{axis}[
axis line style={semithick},
width=0.94\columnwidth,
height=6.5cm,
ylabel style={yshift=-3pt, font=\scriptsize} ,
xlabel style={yshift=3pt, font=\footnotesize} ,
y tick label style={font=\footnotesize} , 
x tick label style={font=\footnotesize}, 
title style={font=\footnotesize},
ylabel style={yshift=-2pt},
  ytick distance=0.2, 
tick pos=left,
x grid style={darkgray176},
xlabel={$Q$ (number of antennas)},
xmajorgrids,
xmin=17, xmax=83,
xminorgrids,
xtick style={color=black},
y grid style={darkgray176},
ylabel={NMSE},
ymajorgrids,
ymin=0.143710783012801, ymax=0.4,
ytick={0.2,0.25,0.3,0.35,0.4}, 
yminorgrids,
ytick style={color=black},
legend style={at={(1,1)}, anchor=north east, legend cell align=left, align=left, draw=white!15!black,font=\scriptsize,opacity=0.99, row sep=-3pt,inner sep=1pt},
legend image post style={xscale=0.63,yscale=0.77},
legend columns=6,
transpose legend,
]

\addplot [semithick, steelblue31119180, mark=*, mark size=1.9, mark options={solid}]
table {%
20 0.192646808944643
30 0.181101317922742
40 0.175534146371809
50 0.17210221397645
60 0.169822792904592
70 0.168349778245402
80 0.16677817607422
};
\addlegendentry{CL-Hyb-bt $K=20$}

\addplot [thick, steelblue31119180,  mark=square*, mark size=1.9, mark options={solid}]
table {%
20 0.258402621639376
30 0.241067197503111
40 0.232573247599855
50 0.22715785877979
60 0.22387183341594
70 0.221811532212556
80 0.219232778579905
};
\addlegendentry{$K=30$}
\addplot [thick, steelblue31119180, mark=triangle*, mark size=2.4, mark options={solid}]
table {%
20 0.340790720364443
30 0.316222208931332
40 0.303615300363223
50 0.29598911775393
60 0.291205816625278
70 0.287661737586766
80 0.284225362062344
};
\addlegendentry{$K=40$}
\addplot [semithick, forestgreen4416044, mark=*, mark size=1.9, mark options={fill=none,solid}]
table {%
20 0.191854548419014
30 0.177232390756661
40 0.170974412265132
50 0.167448674157034
60 0.165310703149715
70 0.164168902933728
80 0.162835681333449
};
\addlegendentry{CL-MP $K=20$}

\addplot [semithick, forestgreen4416044,  mark=square*, mark size=1.9, mark options={fill=none,solid}]
table {%
20 0.274636032858597
30 0.248532468758153
40 0.23643879039714
50 0.229277430841471
60 0.225139974404685
70 0.222641038443056
80 0.219693730210649
};
\addlegendentry{$K=30$}

\addplot [semithick, forestgreen4416044, mark=triangle*, mark size=2.5, mark options={fill=none,solid}]
table {%
20 0.389446209846489
30 0.34682906949594
40 0.32609344904037
50 0.314353992075277
60 0.307092994188484
70 0.30182809478466
80 0.297398423440109
};
\addlegendentry{$K=40$}

\end{axis}

\end{tikzpicture}}
\vspace{-2pt}
\caption{Channel estimation NMSE vs. the number of antennas $Q$ for different numbers of active devices $K$ when $M=30$.  Number of MC trials: $10000$; stopping tolerance: $\delta = 10^{-2}$.} \label{fig:NMSE}
\end{figure} 

Fig.~\ref{fig:NMSE} shows the channel estimation performance for a pilot length of $M=30$. Here, we only compare  CL-Hyb-bt and CL-MP. As can be observed, only in the settings where CL-MP provides slightly better AD performance ($K=20$ and $M=30$ as shown in the middle panel of Fig.~\ref{fig:AD1}), its channel estimation performance is slightly better than that of CL-Hyb-bt. For $K=40$, CL-Hyb-bt provides clearly improved performance while for $K=30$ the differences in channel estimation performance are marginal. We have not included the CWO, CLMM-2, or EM approaches in this figure, as their curves overlap with that of CL-Hyb-bt.

\section{Conclusions}  \label{sec:concl}

We developed three parallel SCA algorithms for covariance learning in the M-SBL framework. CL-SCA preserves the convex part of the objective and linearizes the concave part, yielding closed-form parallel coordinate updates. CL-FP is derived from a projected fixed-point form of the KKT conditions and corresponds to SCA with a strictly convex quadratic surrogate. For both, we established sufficient decrease, an $\mathcal{O}(T^{-1/2})$ best-iterate stationarity rate, and, under standard step-size conditions, stationarity of all limit points. Furthermore,  we derived a backtracking rule that retains this rate without step-size tuning. Local analysis showed that CL-SCA is more stable in the transient phase and CL-FP is faster near stationarity, which motivated the hybrid CL-Hyb.

Experiments on JADCE in GFRA showed that the proposed methods match the detection performance of CWO at substantially lower runtime and are more than an order of magnitude faster than M-SBL. CL-Hyb-bt converged fastest and gave the best accuracy--complexity trade-off among the CL methods.

Future work includes joint estimation of $\gam$ and $\sigma^2$, robust CL-SCA/-FP extensions to  non-Gaussian losses, and unfolding the iterations into a network.

\appendix 
\subsection{Proof of Lemma~\ref{lem:fixed_point}}\label{app:lem:fixed_point}

In \cite[Eq.~(27)]{yardibi2010source}, it was shown that if
$\gamma_n >0$, then setting $\nabla_n \ell(\gam) = 0$
yields the estimating equation
\begin{align} \label{eq:gamma_i_sol0}
\gamma_n  = H_n(\gam),  \quad  H_n(\gam)
&= \frac{ \tildeb_n^{\hop} \SCM \tildeb_n}{  (\a_n^{\hop} \tildeb_n)^2}
 -  \frac{1}{ \a_n^{\hop} \tildeb_n} .
\end{align}
Applying the Sherman-Morrison formula\footnote{
$(\A + \mathbf{u} \mathbf{v}^\hop)^{-1}=\A^{-1} - \A^{-1} \mathbf{u} \mathbf{v}^\hop \A^{-1}/(1 + \mathbf{v}^\hop \A^{-1} \mathbf{u})$} 
to $\M_{\setminus n}= \M - \gamma_n \a_n \a_n^{\hop}$, we notice that vectors $\tildeb_n$ and $\b_n$ are parallel: 
\beq \label{eq:tildeb}
\tildeb_n 
=  \b_n / (1-\gamma_n d_n).
\eeq 
It follows from \eqref{eq:tildeb} that 
  \beq 
  \a_n^{\hop} \tildeb_n =  
   \frac{  d_n}{1-\gamma_n d_n}
\Rightarrow 
 \frac{1}{ \a_n^{\hop} \tildeb_n}  
  = \frac{1}{d_n} -  \gamma_n .
  \label{eq:gamma_i_sol1_apu_2} 
 \eeq 
Since $\tildeb_n$ and $\b_n$ are parallel, it further  follows that the first term in \eqref{eq:gamma_i_sol0} can be equivalently written as
\beq \label{eq:gamma_i_sol1_apu} 
 \frac{ \tildeb_n^{\hop} \SCM \tildeb_n}{  (\a_n^{\hop} \tildeb_n)^2 }
 =  \frac{ \b_n^{\hop} \SCM \b_n}{  (\a_n^{\hop} \b_n)^2}  = \frac{e_n}{d_n^2}. 
\eeq
Substituting \eqref{eq:gamma_i_sol1_apu} and \eqref{eq:gamma_i_sol1_apu_2} into \eqref{eq:gamma_i_sol0} gives~\eqref{eq:gamma_i_sol2}.
Then recalling \eqref{eq:grad_ell}, we note that  \eqref{eq:gamma_i_sol2} can be rewritten as  \eqref{eq:H_scaled_gradient}, 
i.e., $[H(\gam)]_+$ is a projected-gradient map with the positive diagonal scaling $1/d_n^2=(\a_n^\hop\b_n)^{-2}$. 

Suppose first that $\gam$ satisfies the KKT conditions \eqref{eq:KKT}. If $\gamma_n>0$, then \eqref{eq:H_scaled_gradient} gives $H_n(\gam)=\gamma_n>0$, so $[H_n(\gam)]_+=\gamma_n$. If $ \gamma_n=0$, then $H_n(\gam)=-\nabla_n \ell(\gam)/d_n^2\le0$, so $[H_n(\gam)]_+=0=\gamma_n$. Hence \eqref{eq:lik_ver2} holds. Conversely, suppose $\gam=[H(\gam)]_+$. If $\gamma_n>0$, the projection must be inactive, so $\gamma_n=H_n(\gam)$, and  \eqref{eq:H_scaled_gradient} gives $\nabla_n \ell(\gam)=0$. If $\gamma_n=0$, then $[H_n(\gam)]_+=0$ means $H_n(\gam)\le0$, which
by~\eqref{eq:H_scaled_gradient} is exactly $\nabla_n\ell(\gam)\ge0$. Hence the KKT conditions hold,
completing the proof.

\subsection{Proof of Theorem~\ref{th:SCA}}\label{app:th:SCA}

The convex term $g(\gam)=\tr(\M^{-1}\S)$ can expressed as
\begin{align}
  \gfnc(\gam)
  &= \Tr\left(\big(\gamma_n\,\bm{a}_n\bm{a}_n^\Hmat + \bm{\Sigma}_{\backslash n}\big)^{-1}\,\bm{S} \right) \notag \\
  &= \Tr(\bm{\Sigma}_{\backslash n}^{-1}\bm{S}) - \frac{\gamma_n}{1 + \gamma_n\,\bm{a}_n^\Hmat \bm{\Sigma}_{\backslash n}^{-1} \bm{a}_n}\Tr \big(\bm{\Sigma}_{\backslash n}^{-1}\bm{a}_n\bm{a}_n^\Hmat\bm{\Sigma}_{\backslash n}^{-1}\bm{S}\big) \notag\\
  &= \Tr\big(\bm{\Sigma}_{\backslash n}^{-1}\bm{S}\big) - \frac{\gamma_n \bm{a}_n^\Hmat\bm{\Sigma}_{\backslash n}^{-1}\bm{S}\bm{\Sigma}_{\backslash n}^{-1}\bm{a}_n}{1 + \gamma_n\,\bm{a}_n^\Hmat \bm{\Sigma}_{\backslash n}^{-1} \bm{a}_n} \notag  \\
  &= \Tr\big(\bm{\Sigma}_{\backslash n}^{-1}\bm{S}\big) - \frac{\gamma_n \c_n^\hop \S \c_n}{1 + \gamma_n\,\a_n^\Hmat \c_n}\label{eq:f_cvx_gammai}
\end{align}
where $\c_n$ is defined in  \eqref{eq:b} (and also recall \eqref{eq:tildeb}). Note that in the second identity, we used  the Sherman-Morrison formula. 
Then,  by taking the derivative of \eqref{eq:f_cvx_gammai}, we can write 
$e_n(\gam)$ as
\begin{align}
e_n(\gam)=  - \nabla_{n}\gfnc(\gam)  &= \frac{ \c_n^\hop \S \c_n}{(1 + \gamma_n\,\a_n^\Hmat \c_n)^2} 
   \label{eq:partial_gammai_ncvx2_apu}
\end{align} 
Let $\b_n^{\iter} =  \big(\M^{\iter} \big)^{-1} \a_n$  and $\c_n^{\iter} $ be defined  as
\begin{align} 
\c_n^{\iter} &=\big(\M^{\iter}_{\setminus n} \big)^{-1} \a_n  = (\M^{\iter} - \gamma_n^{\iter}  \a_n \a_n^{\hop})^{-1}  \a_n
\end{align} 
i.e., as values of vectors \eqref{eq:b} evaluated  at  $\gam = \gam^{\iter}$ where $\M^{\iter} = \A \diag(\gam^{\iter}) \A^{\hop} + \sigma^2 \mathbf{I}$.
With the above notations and with the help of \eqref{eq:f_cvx_gammai} and \eqref{eq:partial_gammai_ncvx} we can write \eqref{eq:tilde_f} compactly as 
\begin{align}  \label{eq_tildef_i}
\ell_n(\gamma_n \mid \gam^{\iter}) 
&= - \frac{\gamma_n [\c_n^{\iter}]^\Hmat \S \c_n^{\iter}}{1 + \gamma_n\,\a_n^{\hop} \c_n^{\iter}} + \gamma_n d_n^k
\end{align}
where $d_n^k=d_n(\gam^k)= \a_n^\hop \b_n^k$. Above we ignored additive constants in \eqref{eq_tildef_i}  not dependent on $\gamma_n$.

 Taking the derivative of $\ell_n(\gamma_n \mid \gam^{\iter})$ in \eqref{eq_tildef_i} and setting it to zero yields the equation
\begin{align*}
0 &=  \nabla_{n} \ell_n(\gamma_n \mid \gam^{\iter} ) = - \frac{[\c_n^{\iter}]^\Hmat \S \c_n^{\iter}}{ (1 + \gamma_n\a_n^{\hop} \c_n^{\iter})^2} + d_n^{\iter} , 
\end{align*} 
which has a single root: 
\beq \label{eq:gamma_i_update}
\tilde \gamma_n(\gam^k)  = \sqrt{\frac{ [\c_n^{\iter}]^{\hop} \S \c_n^{\iter}}{d_n^{\iter} (\bm{a}_n^\Hmat \c_n^{\iter})^2}} - \frac{1}{\bm{a}_n^\Hmat \c_n^{\iter}} . 
\eeq
This expression is not convenient due to the term $\c_n^{\iter}$ as it requires the computation of all the inverses $\big(\bm{\Sigma}_{\backslash n}^k\big)^{-1}$ for each coordinate $n=1,\dots,N$. 

Equations \eqref{eq:gamma_i_sol1_apu_2} and \eqref{eq:gamma_i_sol1_apu} hold for all $\gam$, so also at $\gam = \gam^k$, meaning that
\begin{align}
 \frac{1}{ \a_n^{\hop} \tildeb_n^{\iter}}  &= \frac{1}{d_n^{\iter}} -  \gamma_n^{\iter} \quad \mbox{and} \quad  \frac{[\tildeb_n^{\iter}]^{\hop} \SCM \tildeb_n^{\iter}}{  (\a_n^{\hop} \tildeb_n^{\iter})^2} = \frac{e_n^k}{(d_n^k)^2} 
  \label{eq:gamma_i_sol1_apu_1} 
\end{align} 
where $e_n^k=e_n(\gam^k)$ by \eqref{eq:dn_and_en} is $e_n^k=[\b_n^{\iter}] ^{\hop} \SCM \b_n^{\iter}$. 
Then  substituting  \eqref{eq:gamma_i_sol1_apu_1} into \eqref{eq:gamma_i_update} allows us to write the root in an equivalent form 
\begin{align} \label{eq:60}
\tilde{\gamma}_n(\gam^k)
&= \gamma_n^{\iter}  + \sqrt{\frac{e_n^{\iter}}{ (d_n^k)^3}}  -   \frac{1}{d_n^{\iter}} 
= \gamma_n^{\iter} + \frac{1}{d_n^k} \Bigg( \sqrt{\frac{e_n^k}{d_n^k}} - 1 \Bigg)
\end{align}
If  $\tilde \gamma_n \ge 0$, the unconstrained minimizer is feasible (optimal), while if $\tilde \gamma_n < 0$, then $\nabla_{n} \ell_n(\gamma_n) > 0$ for all $\gamma_n > \tilde \gamma_n$, in particular for all $\gamma_n \ge 0$. Thus the solution is 
$\hat \gamma_n(\gam^k)= \max( \tilde{\gamma}_n(\gam^k),0)$ as given in  \eqref{eq:SCA_optim_gammai}. 

\subsection{Proof of Lemma~\ref{lem:g_first_order}}\label{app:lem:g_first_order}

Recall from \eqref{eq:grad_ell} that  $\nabla_n\ell(\gam)=d_n-e_n$  while the update map of Theorem~\ref{th:SCA} can be expressed as  $$\hat\gamma_n(\gam) =\bigl[\gamma_n+\sqrt{e_n/d_n^3}-1/d_n\bigr]_+.$$ Fix $n$ and suppose $\hat\gam(\gam)=\gam$. If $\gamma_n>0$, the projection in $\hat\gamma_n$ is inactive, so $\sqrt{e_n/d_n^3}=1/d_n$, i.e., $e_n=d_n$, and thus $\nabla_n\ell(\gam)=0$. If $\gamma_n=0$, then $\hat\gamma_n=0$ requires $\sqrt{e_n/d_n^3}\le1/d_n$, i.e., $e_n\le d_n$, and hence $\nabla_n\ell(\gam)\ge0$. Thus \eqref{eq:KKT} holds. Conversely, if $\gamma_n>0$ and
$\nabla_n\ell(\gam)=0$, then $e_n=d_n$ and \eqref{eq:SCA_optim_gammai} returns $\hat\gamma_n=\gamma_n$.  If $\gamma_n=0$ and $\nabla_n\ell(\gam)\ge0$, then $e_n\le d_n$ and \eqref{eq:SCA_optim_gammai} returns $\hat\gamma_n=0$. Hence $\hat\gam(\gam)=\gam$.

\subsection{Proof of Lemma~\ref{lem:g_constants}} \label{app:lem:g_constants}

(i) For $\gam\in\mathcal K_0$, dropping the nonnegative trace term gives $\log|\M|\le \ell(\gam)\le \ell(\gam^0)$. Since every eigenvalue of $\M$ is at least $\sigma^2$, $|\M|\ge\lambda_{\max}(\M)\,\sigma^{2(M-1)}$, and the bound $\lambda_{\max}(\M)\le\bar\lambda_0$ follows.  Moreover, since $\a_n^\hop(\A\diag(\gam)\A^\hop)\a_n=\sum_{m}\gamma_m|\a_m^\hop\a_n|^2\ge\gamma_n\|\a_n\|^4$,
the Rayleigh quotient gives
\[
  \lambda_{\max}(\M)\;\ge\;\frac{\a_n^\hop\M\a_n}{\|\a_n\|^2}
  \;\ge\;\sigma^2+\gamma_n\|\a_n\|^2 .
\]
Thus, $0\le\gamma_n\le(\bar\lambda_0-\sigma^2)/\|\a_n\|^2$ on $\mathcal K_0$, so $\mathcal K_0$ is bounded, and it is closed by continuity of $\ell$, hence compact. Finally, $\a_n^\hop\M^{-1}\a_n \ge\|\a_n\|^2/\lambda_{\max}(\M)\ge a_{\min}^2/\bar\lambda_0$ for all $\gam\in\mathcal K_0$ and all $n$.
 
(ii) Differentiating \eqref{eq:grad_ell} once more, with $\bm{\Psi}=\A^\hop\M^{-1}\A$ and $\bm{\Phi}=\A^\hop\M^{-1}\SCM\M^{-1}\A$ (both positive semidefinite, $[\bm{\Psi}]_{nn}=d_n$, $[\bm{\Phi}]_{nn}=e_n$), a direct computation gives
\begin{equation}
    [\nabla^2\ell(\gam)]_{nm}
    =-|[\bm{\Psi}]_{nm}|^2+2\,\mathrm{Re}\{[\bm{\Psi}]_{nm}[\bm{\Phi}]_{nm}^\ast\},
    \label{eq:g_hessian}
\end{equation}
i.e., $\nabla^2\ell =2 \mathrm{Re}(\bm{\Psi}\odot\bm{\Phi}^\ast)-\bm{\Psi}\odot\bm{\Psi} ^\ast$. By the triangle inequality and $\|\mathrm{Re}(\bm C)\|_2\le\|\bm C\|_2$ (note that the spectral norm is invariant to the sign of its argument, so the minus sign in front of $\bm{\Psi}\odot\bm{\Psi}^\ast$ is irrelevant), 
\begin{equation}     \label{eq:g_hessian_triangle}
    \|\nabla^2 \ell(\gam)\|_2     \le     2\,\|\bm{\Psi}\odot\bm{\Phi}^\ast\|_2  + \|\bm{\Psi}\odot\bm{\Psi}^\ast\|_2 .
\end{equation}
Each term is now bounded by the Schur-multiplier inequality: for Hermitian $\bm C$ and positive semidefinite $\bm B$, $\|\bm C\odot\bm B\|_2\le(\max_n[\B]_{nn})\|\bm C\|_2$. Applying it with $(\bm C,\bm B)=(\bm M,\bm N^\ast)$ and $(\bm C,\bm B)=(\bm{\Psi},\bm{\Phi}^\ast)$ (both $\bm{\Phi}^\ast$ and $\bm{\Psi}^\ast$ are positive semidefinite), and using $\| \mathbf{C} \mathbf{B}\|_2 \leq \| \mathbf{C}\|_2 \| \mathbf{B}\|_2$ and 
$\M\succeq\sigma^2\I$, which holds for \emph{every} $\gam\in\mathbb{R}_{\ge0}^N$, and gives
\begin{equation*}
    \|\bm{\Psi}\|_2\le\frac{\|\A\|_2^2}{\sigma^2},
    \quad
    [\bm{\Psi}]_{nn}\le\frac{\|\a_n\|^2}{\sigma^2},
    \quad
    [\bm{\Phi}]_{nn}\le\frac{\|\SCM\|_2\|\a_n\|^2}{\sigma^4},
\end{equation*}
yields
\begin{equation*}
    \|\nabla^2\ell(\gam)\|_2
    \le
    2\,\frac{\|\SCM\|_2\,a_{\max}^2}{\sigma^4}\cdot
    \frac{\|\A\|_2^2}{\sigma^2}
    +\frac{a_{\max}^2}{\sigma^2}\cdot\frac{\|\A\|_2^2}{\sigma^2}
    =L_\nabla
\end{equation*}
uniformly on $\mathbb{R}_{\ge0}^N$.

\subsection{Proof of Theorem~\ref{th:g_rate}} \label{app:th:g_rate}

First we claim that for every $\gam\in\mathbb R_{\ge0}^N$, with
$\p=\hat\gam(\gam)-\gam$ and $d_n=d_n(\gam)$,
\beq 
\begin{aligned}  \label{eq:th:g_rate_apu}
    &\nabla \ell(\gam)^\top \p \le  -\sum_{n=1}^{N} d_n^{2} p_n^{2}  
    \Leftrightarrow \sum_{n=1}^N ( \nabla_n \ell(\gam) p_n - d_n^2 p_n^2) \leq 0
\end{aligned}
\eeq
Indeed, write $e_n=e_n(\gam)\ge0$, $r_n=e_n/d_n$ (recall $d_n>0$), and $s_n=1+d_n p_n$. Since $\nabla_n \ell(\gam)=d_n-e_n=d_n(1-r_n)$  and $d_n p_n=s_n-1$,  each summand  in \eqref{eq:th:g_rate_apu} satisfies
\begin{align*}
  (d_n-e_n)p_n+d_n^2p_n^2 &=(1-r_n)(s_n-1)+(s_n-1)^2 \\
  &=(s_n-1)(s_n-r_n),
\end{align*}
which is nonpositive whenever $s_n$ lies between $r_n$ and $1$.  From the update  $\hat\gamma_n=[\gamma_n+\sqrt{e_n/d_n^3}-1/d_n]_+$, we notice that $p_n=\hat\gamma_n-\gamma_n
=\max\{\sqrt{e_n/d_n^3}-1/d_n,\,-\gamma_n\}$. Multiplying by $d_n>0$ and
adding one gives
\begin{equation*}
  s_n=1+d_n p_n=\max\bigl\{\sqrt{r_n}, 1-\gamma_n d_n  \bigr\},
\end{equation*}
where the first branch corresponds to an inactive projection and the second to $\hat\gamma_n=0$. If $s_n=\sqrt{r_n}$, then $s_n$ is the geometric mean of $r_n$ and $1$ and hence lies between them. Otherwise
$\sqrt{r_n}<s_n=1-\gamma_n d_n \le1$, so that $r_n\le\sqrt{r_n}<s_n\le1$.
Summing over $n$ proves~\eqref{eq:th:g_rate_apu}.  Furthermore, for $\gam^k \in\mathcal K_0$,  invoking $d_n(\gam^k)\ge\underline d$ by  Lemma~\ref{lem:g_constants} gives
$\nabla\ell(\gam^k)^\top\p^k \le-\underline{d}^{2}\|\p^k\|_2^2$. 

Assume $\gam^k \in\mathcal K_0$. Since $\gam^{k+1}=\gam^k+\eta^k \p^k$, with $\mathbf{p}^k = \hat{\gam}(\gam^k) - \gam^k$, lies on the segment between $\gam^k$ and $\hat\gam(\gam^k)$,  the quadratic upper bound \eqref{eq:quad_upperbound_L} gives
\begin{align}
    \ell(\gam^{k+1})
     &\le   \ell(\gam^k) +\eta^k \nabla \ell(\gam^k)^\top \p^k  +\frac{L_\nabla}{2}(\eta^k)^2\|\p^k\|_2^2   \notag \\ 
     &\le  \ell(\gam^k) -\eta^k \left(\underline{d}^2-\frac{L_{\nabla}\eta^k}{2}\right)  \|\p^k\|_2^2 \notag\\
    &\le  \ell(\gam^k)  -\frac{\underline d^2}{2}\eta^k \|\p^k\|_2^2,    \label{eq:app_sufficient_decrease}
\end{align}
where the 2nd inequality uses previous assertion, and the last uses $\eta^k \le \underline{d}^2/L_{\nabla}$. Since $\gam^0 \in\mathcal K_0$, \eqref{eq:app_sufficient_decrease} also proves inductively that every iterate remains in $\mathcal K_0$. 

Note that summing \eqref{eq:g_sufficient_decrease} over $k=0,\ldots,T-1$ and telescoping yields
\begin{equation*}
    \frac{\underline d^{\,2}}{2}
    \sum_{k=0}^{T-1}\eta^k\|\bm p^k\|_2^2 \le \ell(\gam^0)- \ell(\gam^T)   \le \ell(\gam^0)- \ell_{\inf},
\end{equation*} 
and bounding the weighted average from below by its smallest term gives \eqref{eq:g_rate}. For $\gam^0=\bm0$, note $\ell(\bm0)=\tr(\SCM)/\sigma^2+M \log\sigma^2$ and $\ell_{\inf}\ge M \log\sigma^2$. Thus $\ell(\bm 0)- \ell_{\inf}\le\tr(\SCM)/\sigma^2$. If $\inf_k\eta^k>0$, the summability of $\eta^k\|\bm p^k\|_2^2$ forces $\|\bm p^k\|_2\to0$. Every accumulation point $\gam^\star$ then satisfies $\hat\gam(\gam^\star)=\gam^\star$ by continuity of the update map, and is KKT-stationary by Lemma~\ref{lem:g_first_order}.

\subsection{Proof of Theorem~\ref{th:fp_rate}} \label{app:th:fp_rate}

We claim that for every $\gam\in\mathbb{R}_{\ge0}^N$, with
$\bm q=[H(\gam)]_+-\gam$ and $g_n=[\nabla\ell(\gam)]_n=d_n-e_n$,
\begin{equation}
    \nabla\ell(\gam)^\top\bm q
    \le-\sum_{n=1}^N d_n^2\,q_n^2 .
    \label{eq:fp_descent_direction}
\end{equation} 
Fix $n$. If the projection is inactive, then $q_n=H_n(\gam)-\gamma_n=-g_n/d_n^2$ by \eqref{eq:H_scaled_gradient}, so $g_n q_n=-g_n^2/d_n^2=-d_n^2 q_n^2$, with equality. If the projection
is active and $\gamma_n>0$, then $	q_n=-\gamma_n$ and $H_n(\gam)<0$, i.e., $g_n>\gamma_nd_n^2>0$, hence $g_n q_n=-\gamma_ng_n<-\gamma_n^2d_n^2=-d_n^2 q_n^2$. If the projection is active and $\gamma_n=0$, then $q_n=0$ and the term vanishes. Summing over $n$ proves \eqref{eq:fp_descent_direction}, and for $\gam\in\mathcal K_0$ invoking $d_n(\gam)\ge\underline d$ gives $\nabla\ell(\gam)^\top\bm q\le-\underline d^{\,2}\|\bm q\|_2^2$. 

Proceed by induction, assume $\gam^k\in\mathcal K_0$ (true for $k=0$). Since $\gam^{k+1}$ lies on the segment $[\gam^k,[H(\gam^k)]_+]\subset\mathbb{R}_{\ge0}^N$, on which $\nabla\ell$ is $L_\nabla$-Lipschitz by Lemma~\ref{lem:g_constants}(ii), the descent lemma and \eqref{eq:fp_descent_direction} with
$d_n(\gam^k)\ge\underline d$ give
\begin{align*}
    \ell(\gam^{k+1})
    &\le\ell(\gam^k)
    +\eta^k\nabla\ell(\gam^k)^\top\bm q^k
    +\frac{L_\nabla}{2}(\eta^k)^2\|\bm q^k\|_2^2\\
    &\le\ell(\gam^k)
    -\eta^k\Bigl(\underline d^{\,2}
    -\frac{L_\nabla\,\eta^k}{2}\Bigr)\|\bm q^k\|_2^2 \\
    &\le\ell(\gam^k)
    -\frac{\underline d^{\,2}}{2}\eta^k\|\bm q^k\|_2^2 ,
\end{align*}
where the last step uses $\eta^k\le\underline d^{\,2}/L_\nabla$. In particular $\ell(\gam^{k+1})\le\ell(\gam^k)\le\ell(\gam^0)$, so $\gam^{k+1}\in\mathcal K_0$, closing the induction and proving~\eqref{eq:fp_sufficient_decrease}. The rest follows as in proof of Theorem~\ref{th:g_rate}



\begin{thebibliography}{10}
\providecommand{\url}[1]{#1}
\csname url@samestyle\endcsname
\providecommand{\newblock}{\relax}
\providecommand{\bibinfo}[2]{#2}
\providecommand{\BIBentrySTDinterwordspacing}{\spaceskip=0pt\relax}
\providecommand{\BIBentryALTinterwordstretchfactor}{4}
\providecommand{\BIBentryALTinterwordspacing}{\spaceskip=\fontdimen2\font plus
\BIBentryALTinterwordstretchfactor\fontdimen3\font minus
  \fontdimen4\font\relax}
\providecommand{\BIBforeignlanguage}[2]{{%
\expandafter\ifx\csname l@#1\endcsname\relax
\typeout{** WARNING: IEEEtran.bst: No hyphenation pattern has been}%
\typeout{** loaded for the language `#1'. Using the pattern for}%
\typeout{** the default language instead.}%
\else
\language=\csname l@#1\endcsname
\fi
#2}}
\providecommand{\BIBdecl}{\relax}
\BIBdecl

\bibitem{duarte2011structured}
M.~F. Duarte and Y.~C. Eldar, ``Structured compressed sensing: From theory to
  applications,'' \emph{{IEEE} Trans. Signal Process.}, vol.~59, no.~9, pp.
  4053--4085, 2011.

\bibitem{tropp2006algorithms}
J.~A. Tropp, A.~C. Gilbert, and M.~J. Strauss, ``Algorithms for simultaneous
  sparse approximation. {Part I:} greedy pursuit,'' \emph{Signal Processing},
  vol.~86, pp. 572--588, 2006.

\bibitem{tropp:2006}
J.~A. Tropp, ``Algorithms for simultaneous sparse approximation. {Part II:}
  convex relaxation,'' \emph{Signal Processing}, vol.~86, pp. 589--602, 2006.

\bibitem{chen_huo:2006}
J.~Chen and X.~Huo, ``Theoretical results on sparse representations of
  multiple-measurement vectors,'' \emph{{IEEE} Trans. Signal Process.},
  vol.~54, no.~12, pp. 4634--4643, 2006.

\bibitem{eldar_rauhut:2010}
Y.~C. Eldar and H.~Rauhut, ``Average case analysis of multichannel sparse
  recovery using convex relaxation,'' \emph{{IEEE} Trans. Inf. Theory},
  vol.~56, no.~1, pp. 505--519, 2010.

\bibitem{blanchard2014greedy}
J.~D. Blanchard, M.~Cermak, D.~Hanle, and Y.~Jin, ``Greedy algorithms for joint
  sparse recovery,'' \emph{{IEEE} Trans. Signal Process.}, vol.~62, no.~7, pp.
  1694 -- 1704, 2014.

\bibitem{tipping2001sparse}
M.~E. Tipping, ``Sparse bayesian learning and the relevance vector machine,''
  \emph{J. Mach. Learn. Res.}, vol.~1, no. Jun, pp. 211--244, 2001.

\bibitem{wipf2004sparse}
D.~P. Wipf and B.~D. Rao, ``Sparse bayesian learning for basis selection,''
  \emph{{IEEE} Trans. Signal Process.}, vol.~52, no.~8, pp. 2153--2164, 2004.

\bibitem{scutari2013decomposition}
G.~Scutari, F.~Facchinei, P.~Song, D.~P. Palomar, and J.-S. Pang,
  ``Decomposition by partial linearization: Parallel optimization of
  multi-agent systems,'' \emph{{IEEE} Trans. Signal Process.}, vol.~62, no.~3,
  pp. 641--656, 2013.

\bibitem{cotter2005sparse}
S.~F. Cotter, B.~D. Rao, K.~Engan, and K.~Kreutz-Delgado, ``Sparse solutions to
  linear inverse problems with multiple measurement vectors,'' \emph{{IEEE}
  Trans. Signal Process.}, vol.~53, no.~7, pp. 2477--2488, 2005.

\bibitem{wipf2007empirical}
D.~P. Wipf and B.~D. Rao, ``An empirical {B}ayesian strategy for solving the
  simultaneous sparse approximation problem,'' \emph{{IEEE} Trans. Signal
  Process.}, vol.~55, no.~7, pp. 3704--3716, 2007.

\bibitem{esfandiari2026mm}
M.~Esfandiari, E.~Ollila, and D.~P. Palomar, ``Majorization-minimization
  covariance learning algorithm for sparse signal recovery,'' in \emph{2026
  34th European Signal Processing Conference ({EUSIPCO})}, 2026, pp. 1--5.

\bibitem{al2017gamp}
M.~Al-Shoukairi, P.~Schniter, and B.~D. Rao, ``A {GAMP}-based low complexity
  sparse {B}ayesian learning algorithm,'' \emph{{IEEE} Trans. Signal Process.},
  vol.~66, no.~2, pp. 294--308, 2017.

\bibitem{hashemi2022joint}
A.~Hashemi, C.~Cai, Y.~Gao, S.~Ghosh, K.-R. M{\"u}ller, S.~S. Nagarajan, and
  S.~Haufe, ``Joint learning of full-structure noise in hierarchical bayesian
  regression models,'' \emph{{IEEE} Trans. Med. Imag.}, vol.~43, no.~2, pp.
  610--624, 2022.

\bibitem{marata2025activity}
L.~Marata, E.~Ollila, and H.~Alves, ``Activity detection for massive random
  access using covariance-based matching pursuit,'' \emph{{IEEE} Trans. Veh.
  Technol.}, vol.~74, no.~11, pp. 17\,292--17\,303, 2025.

\bibitem{pote2025theory}
R.~R. Pote and B.~D. Rao, ``Theory and practice of light-weight sequential sbl
  algorithm: An alternative to {OMP},'' \emph{{IEEE} Trans. Signal Process.},
  2025.

\bibitem{liu2018sparse}
L.~Liu, E.~G. Larsson, W.~Yu, P.~Popovski, C.~Stefanovic, and E.~De~Carvalho,
  ``Sparse signal processing for grant-free massive connectivity: A future
  paradigm for random access protocols in the internet of things,''
  \emph{{IEEE} Signal Process. Mag.}, vol.~35, no.~5, pp. 88--99, 2018.

\bibitem{liu2018massive}
L.~Liu and W.~Yu, ``Massive connectivity with massive {MIMO}-{P}art {I}: Device
  activity detection and channel estimation,'' \emph{{IEEE} Trans. Signal
  Process.}, vol.~66, no.~11, pp. 2933--2946, 2018.

\bibitem{chen2018sparse}
Z.~Chen, F.~Sohrabi, and W.~Yu, ``Sparse activity detection for massive
  connectivity,'' \emph{{IEEE} Trans. Signal Process.}, vol.~66, no.~7, pp.
  1890--1904, 2018.

\bibitem{fengler2021non}
A.~Fengler, S.~Haghighatshoar, P.~Jung, and G.~Caire, ``Non-{B}ayesian activity
  detection, large-scale fading coefficient estimation, and unsourced random
  access with a massive {MIMO} receiver,'' \emph{{IEEE} Trans. Inf. Theory},
  vol.~67, no.~5, pp. 2925--2951, 2021.

\bibitem{chen2019covariance}
Z.~Chen, F.~Sohrabi, Y.-F. Liu, and W.~Yu, ``Covariance based joint activity
  and data detection for massive random access with massive {MIMO},'' in
  \emph{{ICC} 2019}, 2019, pp. 1--6.

\bibitem{chen2021phase}
------, ``Phase transition analysis for covariance-based massive random access
  with massive {MIMO},'' \emph{{IEEE} Trans. Inf. Theory}, vol.~68, no.~3, pp.
  1696--1715, 2021.

\bibitem{wang2025robust}
X.~Wang, E.~Ollila, and S.~A. Vorobyov, ``Robust activity detection for massive
  random access,'' \emph{{IEEE} Trans. Signal Process.}, vol.~73, pp. 3513 --
  3527, 2025.

\bibitem{tan2025fast}
B.~Tan and Y.~Cui, ``Fast {MLE} and {MAPE}-based device activity detection for
  grant-free access via {PSCA} and {PSCA-Net},'' \emph{{IEEE} Trans. Wireless
  Commun.}, vol.~24, no.~6, pp. 5066--5083, 2025.

\bibitem{ollila2026joint}
E.~Ollila, M.~Esfandiari, and D.~P. Palomar, ``Joint activity detection and
  channel estimation for massive random access using {SBL} and {SCA},'' in
  \emph{2026 {IEEE} {SPAWC}}, Athens, Greece, Sep 6--9, 2026, pp. 1--5.

\bibitem{yardibi2010source}
T.~Yardibi, J.~Li, P.~Stoica, M.~Xue, and A.~B. Baggeroer, ``Source
  localization and sensing: A nonparametric iterative adaptive approach based
  on weighted least squares,'' \emph{{IEEE} Trans. Aerosp. Electron. Syst.},
  vol.~46, no.~1, pp. 425--443, 2010.

\bibitem{wipf2007beamforming}
D.~Wipf and S.~Nagarajan, ``Beamforming using the relevance vector machine,''
  in \emph{Proceedings of the 24th international conference on Machine
  learning}, 2007, pp. 1023--1030.

\bibitem{scutari2018parallel}
G.~Scutari and Y.~Sun, ``Parallel and distributed successive convex
  approximation methods for big-data optimization,'' in \emph{Multi-Agent
  Optimization: Cetraro, Italy 2014}.\hskip 1em plus 0.5em minus 0.4em\relax
  Springer, 2018, pp. 141--308.

\bibitem{nesterov2004introductory}
Y.~Nesterov, \emph{Introductory Lectures on Convex Optimization: A Basic
  Course}.\hskip 1em plus 0.5em minus 0.4em\relax Boston, MA: Kluwer Academic
  Publishers, 2004.

\bibitem{bertsekas1999nonlinear}
D.~P. Bertsekas, \emph{Nonlinear Programming}, 2nd~ed.\hskip 1em plus 0.5em
  minus 0.4em\relax Belmont, MA: Athena Scientific, 1999.

\bibitem{sun2016robust}
Y.~Sun, P.~Babu, and D.~P. Palomar, ``Robust estimation of structured
  covariance matrix for heavy-tailed elliptical distributions,'' \emph{{IEEE}
  Trans. Signal Process.}, vol.~64, no.~14, pp. 3576--3590, 2016.

\end{thebibliography}
\end{document}